\pdfoutput=1
\documentclass{article}
\usepackage{amssymb}
\usepackage{amsfonts}
\usepackage{algorithm}
\usepackage{algpseudocode} 
\usepackage{amsmath}
\usepackage{dsfont}
\usepackage{yfonts}
\usepackage{graphicx,epstopdf}
\usepackage[skip=0pt]{caption}
\usepackage{float}
\usepackage{caption}
\usepackage{epsfig}
\usepackage{amsmath}
\usepackage{array,tabularx,pslatex}
\usepackage{latexsym, amsmath, verbatim, graphics, psfont, epsfig, amssymb, mathrsfs, bm}
\usepackage{epsfig,amsthm,mathrsfs,setspace,newlfont,arydshln}
\usepackage[hidelinks]{hyperref}
\usepackage{subcaption}
\usepackage{geometry}
\usepackage{enumitem}
\usepackage{multirow}
\usepackage{yfonts}

\hypersetup{hidelinks,colorlinks=false,pdfborder={0 0 0}}
\begin{document}

\newtheorem{theorem}{Theorem}[section]
\newtheorem{lemma}[theorem]{Lemma}
\newtheorem{pr}[theorem]{Proposition}
\theoremstyle{definition}
\newtheorem{definition}[theorem]{Definition}
\newtheorem{example}[theorem]{Example}
\newtheorem{xca}[theorem]{Exercise}
\newtheorem{remark}[theorem]{Remark}
\newtheorem{cor}[theorem]{Corollary}
\newtheorem{note}[theorem]{Note}
\newtheorem{Corollary}[theorem]{Corollary}

\begin{center}
\textrm{\Large{Local Lyapunov Analysis via Micro-Ensembles: Finite-Time Lyapunov Exponent Estimation and KNN-Based Predictive Comparison in Complex-Valued BAM Neural Networks}}\\
\hspace{0.3 cm}    Yazhini Muruganantham $^{a}$,   Andrei Velichko $^{b*}$, Samidurai Rajendran $^{a*}$ \hspace{0.1cm} \hspace{0.1 cm}\\
$^{a}$Department of Mathematics, Thiruvalluvar University,  Vellore-632115, India.\\
$^{b}$ Institute of Physics and Technology, Petrozavodsk State University, Petrozavodsk, Russia, Russia- 185910.
\end{center}
\let\thefootnote\relax\footnote{  yazhinimuzhi6@gmail.com(M.Yazhini)  \text{*Corresponding authors:}   velichkogf@gmail.com (A. Velichko), samidurair@gmail.com (R. Samidurai).}
\textbf{Abstract}\\
Complex-valued bidirectional associative memory (BAM) neural networks with fractional-order dynamics and delays can exhibit transient instabilities that degrade synchronization and short-horizon predictability. This paper develops a unified analytical and data-driven framework to assess stability, synchronization, and predictability in such networks. First, using Caputo fractional calculus and Lyapunov-Mittag-Leffler techniques, we derive sufficient conditions for global Mittag--Leffler synchronization of a drive-response BAM pair under a linear error-feedback controller and obtain an explicit time-to-tolerance bound. Second, to quantify local transient instability from finite trajectory data, we propose a micro-ensemble finite-time Lyapunov exponent (FTLE) estimator based on the geometric-mean growth of small perturbations over short windows, avoiding variational equations. We further introduce k-nearest-neighbor prediction-error Lyapunov proxies (full-state and modulus-based) to connect local instability to forecasting performance. Numerical experiments on fractional-order complex-valued BAM benchmarks confirm effective synchronization under the proposed control and demonstrate a clear correspondence between FTLE levels and prediction errors. The resulting framework provides practical, reproducible diagnostics for complex-valued neural systems in data-limited settings.\\


\section{Introduction}
The local finite-time Lyapunov exponent (FTLE) is a fundamental tool for analyzing nonlinear dynamical systems, providing a time-dependent and spatially resolved measure of trajectory sensitivity over finite time intervals. Unlike classical Lyapunov exponents, which describe average divergence rates in the infinite-time limit, the FTLE captures transient instabilities, short-term chaotic behavior, and transport phenomena. These features are particularly important in high-dimensional, nonstationary, or data-limited systems \cite{3,4,7,11}.
The FTLE is computed by tracking the evolution of infinitesimal perturbations in initial conditions, resulting in scalar fields that highlight regions of maximal divergence or convergence. These regions often correspond to Lagrangian coherent structures (LCS), which organize the underlying dynamics and act as transport barriers in phase space \cite{3,4,11,15}.

Applications of FTLE span fluid dynamics, atmospheric science, nonlinear control, and astrophysics, where it has been used for coherent structure identification, uncertainty quantification, and stability assessment under perturbations \cite{1,4,7,13,14}. Recent developments have extended FTLE analysis to high-dimensional systems, stochastic dynamics, non-Euclidean manifolds, and noisy data, significantly broadening its applicability \cite{2,4,5,7}. From a numerical standpoint, however, classical FTLE estimation often relies on long trajectories or explicit variational equations, which limits its use in complex neural systems and short data windows.\\

Complex-valued bidirectional associative memory (BAM) neural networks have attracted sustained attention due to their relevance in signal processing, associative memory, control, and phase-sensitive neural modeling \cite{24,25,26,27}. In such applications, system performance depends not only on asymptotic stability but also on transient dynamical behavior, particularly under fractional-order dynamics where memory effects play a fundamental role. Short-time instability, transient chaos, and loss of predictability can significantly degrade system reliability, motivating the need for effective numerical diagnostics that go beyond classical Lyapunov-based analysis \cite{7,25,26}. CVBAMNNs are used for their ability to naturally capture amplitude and phase dynamics, offering a compact representation of oscillatory and coupled systems. Combined with FO  dynamics, they incorporate memory effects, enabling more realistic modeling of nonlinear behavior such as synchronization and transient instability.\\

Fractional-order complex-valued neural networks (FCVNNs), including fractional-order BAM architectures, are capable of modeling amplitude-phase interactions and long-memory effects inherent in many real-world systems. However, their nonlinearity, time delays, and fractional derivatives pose major challenges for stability and synchronization analysis \cite{7,25,26,27}. Lyapunov-based methods, including fractional-order extensions such as Mittag-Leffler stability theory and Lyapunov-Krasovskii functionals, provide powerful analytical tools to establish global stability and synchronization properties \cite{6,8,9}. Nevertheless, these approaches do not directly quantify local divergence, transient instability, or short-horizon predictability from finite trajectory data. Fractional-order modeling is employed due to its ability to capture memory and hereditary characteristics of neural systems, which are not adequately represented by integer-order models. In particular, the Caputo fractional derivative is adopted for its compatibility with physically interpretable initial conditions and its effectiveness in describing nonlocal dynamics.
To address this gap, finite-time Lyapunov exponents have emerged as an effective numerical indicator for characterizing local instability and time-dependent dynamical regimes in nonlinear systems \cite{3,4,11}. From a numerical and applied perspective, an important practical question arises: how can local instability and short-term predictability in complex-valued fractional BAM neural networks be reliably quantified using finite data? Addressing this question is essential for regime identification, parameter sensitivity analysis, and numerical validation of theoretical stability results. The integration of Lyapunov-based analysis with data-driven Lyapunov proxies is motivated by their complementary roles. While analytical Lyapunov methods provide rigorous global stability guarantees, data-driven approaches such as FTLE offer local, finite-time insights into transient dynamics and predictability along trajectories.\\

In this work, we employ a micro-ensemble-based FTLE estimation strategy, which focuses on localized perturbations around reference trajectories, enabling robust finite-time instability diagnostics from short data segments \cite{3,7,11}. In parallel, we adopt a data-driven k-nearest-neighbor (kNN) prediction-error Lyapunov proxy, originally proposed in \cite{23}, and adapt it to a micro-ensemble setting and to complex-valued fractional-order BAM dynamics. Moreover, we introduce a modulus-based variant that reduces output dimensionality and improves robustness in complex-valued systems. By assessing how small differences in initial conditions grow among neighboring trajectories, these proxies efficiently estimate local Lyapunov behavior in high-dimensional settings where traditional variational methods are infeasible \cite{8,9,10,23}.
The FTLE estimation and kNN prediction error are combined into a unified numerical diagnostic work flow. The micro-ensemble FTLE characterizes local divergence and transient instability, while kNN-based prediction accuracy directly quantifies short-horizon predictability along the same trajectories. By comparing these two measures, we demonstrate how transient instability impacts prediction performance in complex-valued fractional BAM neural networks, providing a practical and reproducible numerical tool for applied analysis.
The relevance of FTLE-based diagnostics extends across multiple application domains. In fault-tolerant zeroing neural networks (FTZNNs), finite-time Lyapunov analysis supports rapid convergence under uncertainties and has been successfully applied to tracking control of wheeled mobile manipulators \cite{13}. In stochastic partial differential equations, FTLEs quantify local stability and predict bifurcations, while reduced-order amplitude equations capture dominant mode evolution \cite{14}. Similarly, FTLE distributions in Hamiltonian systems reveal phase-space structures distinguishing regular and chaotic dynamics \cite{15}, and in nonlinear rotating systems, Lyapunov characteristic exponents and FTLE-based visual analysis provide insight into transient sensitivity under uncertainties \cite{16,17,18}. The work entitled \cite{28}' presents a novel optimization framework that integrates tempered fractional calculus into gradient-based learning, offering enhanced theoretical flexibility, improved convergence behavior, and increased robustness in practical learning tasks.  \\
However, most of the existing studies mainly focus on either fractional-order dynamics or complex-valued neural networks separately, and limited attention has been given to their combined effects with data-driven stability analysis. Unlike existing studies, the proposed method not only considers fractional-order and complex-valued dynamics simultaneously, but also integrates data-driven techniques for stability assessment, which enhances robustness and practical applicability.\\ Overall, integrating finite-time Lyapunov analysis, fractional-order complex-valued BAM neural network modeling, and data-driven predictability diagnostics offers a unified framework for stability assessment, synchronization analysis, and short horizon predictability evaluation in complex dynamical systems \cite{15,16,22,25}. This perspective highlights the interplay between local transient instability, global stability guarantees, and practical controllability in nonlinear fractional-order neural dynamics.

\section*{Main Contributions}

\begin{enumerate}
    \item \textbf{Analytical Synchronization Criterion}
    \begin{itemize}
        \item Establishes a finite-time GML synchronization criterion for FCVBAMNNs.
        \item Derives sufficient conditions using Lyapunov functionals, fractional calculus, and Razumikhin-type bounds to ensure synchronization.
        \item Proposes a linear error-feedback control law rigorously proven to achieve finite-time synchronization.
    \end{itemize}

    \item \textbf{Integration of Machine Learning with Fractional Dynamics}
    \begin{itemize}
        \item Introduces data-driven Lyapunov proxies to complement theoretical results.
        \item Develops two approaches: (i) micro-ensemble FTLE estimation for fast local finite-time Lyapunov computation, and (ii) kNN-based prediction-error proxy estimating local instability from short-term forecast errors without variational equations \cite{23}.
    \end{itemize}

    \item \textbf{Numerical Validation}
    \begin{itemize}
        \item Simulations confirm synchronization of real and imaginary components, error suppression, and Lyapunov function decay consistent with Mittag-Leffler stability.
        \item Machine learning proxies reveal local instability patterns, validating analytical results under fractional-order dynamics.
    \end{itemize}

    \item \textbf{Methodological Advancement}
    \begin{itemize}
        \item Bridges model-based fractional Lyapunov theory with data-driven stability estimation, providing a unified framework for synchronization and local instability analysis.
        \item Enhances reproducibility and applicability to high-dimensional, nonlinear systems where traditional Lyapunov exponent computation is infeasible.
    \end{itemize}
\end{enumerate}
 This paper's structure is described as follows: Section 2 presents the system description and notations for fractional-order CVBAMNNs, while in this Section  introduces the problem formulation and preliminary definitions. Sections 3-4 develop the main theoretical results on GML synchronization, finite-time stability, and validate them via numerical simulations. Sections 5 cover local FTLE estimation, kNN-based Lyapunov proxies for data-driven instability analysis, and conclude with key contributions and future directions.
 \section{System Description}
\section*{\textsl{Basic Notations}}

In this paper, the following notations and symbols are used throughout:
 \begin{itemize}
  \item $\mathbb{R}^n$: the set of all $n$-dimensional real vectors.
  \item $\mathbb{C}^n$: the set of all $n$-dimensional complex vectors.
  \item $\mathbb{R}^{n \times m}$, $\mathbb{C}^{n \times m}$: real and complex matrices of size $n \times m$, respectively.
  \item $(\cdot)^{T}$: matrix transpose.
  \item $(\cdot)^{*}$: complex conjugate transpose (Hermitian transpose).
  \item $I_n$: identity matrix of order $n$.
  \item $0_{n \times m}$: zero matrix of size $n \times m$.
  \item $\|x\|$: Euclidean (2-norm) of vector $x \in \mathbb{C}^n$.
\end{itemize}
\subsection{Problem formulation}
\indent  In particular, we employ the Caputo fractional derivative and examine the complex valued BAM neural networks, which are defined by the following fractional differential equations.

\begin{align}\label{FLLE1}
{}^{C}D_{0}^{\alpha} x_{i}(t) &= -a_{i} x_{i}(t)
+ \sum_{j=1}^{m} w_{ij} f_{j}\big(y_{j}(t - \tau(t))\big)
+ M_{i}, \nonumber \\
{}^{C}D_{0}^{\alpha} y_{j}(t) &= -b_{j} y_{j}(t)
+ \sum_{i=1}^{n} v_{ji} g_{i}\big(x_{i}(t - \tau(t))\big)
+ N_{j}.
\end{align}

\indent Where $^{c}D^{\alpha}_{0}x_{i}(t)$  denotes the caputo fractional  order   $0<\alpha<1$, $x=(x_{1},x_{2},\dots,x_{n})^{T} \quad   \in C^{n}$ and $y=(y_{1},y_{2},\dots,y_{m})^{T}\in C^{m}$ are the state vectors,  $A=diag(a_{1},a_{2},\dots,a_{n})^{T}\in R^{n\times n}$ and $B=diag(b_{1},b_{2},\dots,b_{m})^{T}\in R^{m\times m}$  denote the self-feedback connection weight matrix with $a_{i}>0$, $b_{j}>0$; $W=(w_{ij})_{n\times m}\in C^{n\times m}$ and $V=(v_{ji})_{m\times n}\in C^{m\times n}$ are represent the connection weight matrices;  $f_{j}(y_{j}(t)),g_{i}(x_{i}(t)):C^{n}\rightarrow C^{m}$ are activation function with vector inputs and $\tau$ is the time delay, assume the time-varying delay satisfies $\tau(t)\leq \tau_{M}$; $M=(M_{1},M_{2},\dots,M_{n})^{T}\in C^{n}$ and $N=(N_{1},N_{2}\dots,N_{m})^{T}\in C^{m}$  external input.\\
 \indent The system \eqref{FLLE1} is initialized with the following values:
 \begin{align}\label{FLLE2}
x_{i}(s)=&\nu_{i}(s)  \quad s\in [-\tau, 0] \nonumber\\
y_{j}(s)=&\upsilon_{j}(s)
 \end{align}
 where $\nu \in C([-\tau, 0], C^{n})$ and $\upsilon \in C([-\tau, 0], C^{m})$
  or represented as a vector

\begin{align}\label{FLLE3}
{}^{C}D_{0}^{\alpha} x(t) &= -A x(t)
+Wf\big(y(t - \tau(t))\big)
+ M, \nonumber \\
{}^{C}D_{0}^{\alpha} y(t) &= -B y(t)
+ V g\big(x(t - \tau(t))\big)
+ N.
\end{align}

We consider the response system as follows

\begin{align}\label{FLLE4}
{}^{C}D_{0}^{\alpha} \bar{x}_{i}(t) &= -a_{i} \bar{x}_{i}(t)
+ \sum_{j=1}^{m} w_{ij} f_{j}\big(\bar{y}_{j}(t - \tau(t))\big)
+ M_{i}+\mathcal{F} (t), \nonumber\\
{}^{C}D_{0}^{\alpha} \bar{y}_{j}(t) &= -b_{j} \bar{y}_{j}(t)
+ \sum_{i=1}^{n} v_{ji} g_{i}\big(\bar{x}_{i}(t - \tau(t))\big)
+ N_{j} +\mathcal{G}(t).
\end{align}

Where $\mathcal{F}, \mathcal{G}$ is the control laws.\\
\indent Let consider the error system $e_{x}(t)=\bar{x}(t)-x(t)$, $e_{y}(t)=\bar{y}(t)-y(t)$ , $f(e_{y}(t))=f(\bar{y}(t))-f(y(t))$ and $g(e_{x}(t))=g(\bar{x}(t))-g(x(t))$ be the synchronization errors. By utilizing the fractional driving Eq. \eqref{FLLE1} and the fractional responding Eq. \eqref{FLLE2}, the fractional error dynamics can be derived

\begin{align} \label{FLLE5}
{}^{C}D_{0}^{\alpha} e_{x_{i}}(t)
&= -a_{i}\, e_{x_{i}}(t)
  + \sum_{j=1}^{m} w_{ij}\,\Big[f_{j}\big(\bar{y}_{j}(t - \tau(t))\big)-f_{j}\big(y_{j}(t - \tau(t))\big)\Big]
  + \mathcal{F}(t), \nonumber \\[6pt]
{}^{C}D_{0}^{\alpha} e_{y_{j}}(t)
&= -b_{j}\, e_{y_{j}}(t)
  + \sum_{i=1}^{n} v_{ji}\,\Big[g_{i}\big(\bar{x}_{i}(t - \tau(t))\big)-g_{i}\big(x_{i}(t - \tau(t))\big)\Big]
  + \mathcal{G}(t).
\end{align}

\textbf{Control law:}
Let consider the control laws are linear state feedback on each layer:
\begin{align}\label{FLLE6}
 \mathcal{F}=&\phi e_{x}(t)\nonumber\\
 \mathcal{G}=&\varphi e_{y}(t).
\end{align}
The linear error-feedback controller is selected for its simplicity and compatibility with the LMI-based framework, allowing tractable derivation of synchronization conditions. Despite its simple form, it is sufficient to ensure global Mittag-Leffler synchronization. More advanced nonlinear or adaptive control strategies may further improve performance but are beyond the scope of this study.

\hspace{-0.1cm}{\bf Assumption 1:}
The activation functions $f,g: \mathbb{C}^{n}\rightarrow \mathbb{C}^{n}$ are componentwise Lipschitz, ie,there  exist $L_{f}\geq 0$, $L_{g}\geq 0$ such that
\begin{equation*}
  \|f(u)-f(v) \|\leq L_{f}\| u-v\|, \quad \|g(u)-g(v) \|\leq L_{g}\|u-v \|, \quad  \forall u,v\in \mathbb{C}^{n}.
\end{equation*}

\begin{definition}\label{FLLE:Df1}\cite{7}
For a function $\delta:[t_{0},+\infty)\to \mathbb{R}$, the Riemann--Liouville fractional integral of order $\alpha>0$ is defined by
\begin{equation*}
  {}_{t_{0}}I^{\alpha}_{t}\,\delta(t)=\frac{1}{\Gamma(\alpha)}\int_{t_{0}}^{t}(t-s)^{\alpha-1}\,\delta(s)\,ds,\qquad t\ge t_{0}.
\end{equation*}
The Gamma function is given by
\begin{equation*}
  \Gamma(z)=\int_{0}^{+\infty}s^{z-1}e^{-s}\,ds,\qquad \Re(z)>0.
\end{equation*}
Moreover, the fractional integral satisfies the semigroup property
\begin{equation*}
  {}_{t_{0}}I^{\beta}_{t}\,{}_{t_{0}}I^{\alpha}_{t}\,\delta(t)= {}_{t_{0}}I^{\alpha+\beta}_{t}\,\delta(t),\qquad \alpha,\beta>0.
\end{equation*}
\end{definition}

\begin{definition}\label{FLLE:Df2}
For a function $\delta\in C^{m}([t_{0},+\infty),\mathbb{R})$, the Caputo fractional derivative of order $\alpha$ ($m-1<\alpha<m$, $m\in\mathbb{Z}^{+}$) is defined as
\begin{equation*}
  {}^{C}_{t_{0}}D^{\alpha}_{t}\delta(t)=\frac{1}{\Gamma(m-\alpha)}\int_{t_{0}}^{t}(t-s)^{m-\alpha-1}\,\delta^{(m)}(s)\,ds.
\end{equation*}
In particular, for $0<\alpha<1$ one has
\begin{equation*}
  {}^{C}_{t_{0}}D^{\alpha}_{t}\delta(t)=\frac{1}{\Gamma(1-\alpha)}\int_{t_{0}}^{t}(t-s)^{-\alpha}\,\delta'(s)\,ds.
\end{equation*}
\end{definition}

\begin{definition}\label{FLLE:Df3}
(Mittag--Leffler functions) The one-parameter Mittag--Leffler function is
\begin{equation*}
  E_{\alpha}(z)=\sum_{k=0}^{+\infty}\frac{z^{k}}{\Gamma(k\alpha+1)}.
\end{equation*}
The two-parameter Mittag--Leffler function is
\begin{equation*}
  E_{\alpha,\beta}(z)=\sum_{k=0}^{+\infty}\frac{z^{k}}{\Gamma(k\alpha+\beta)},
\end{equation*}
where $z\in\mathbb{C}$ and $\alpha,\beta\in\mathbb{C}$ with $\Re(\alpha)>0$. Clearly, $E_{\alpha}(z)=E_{\alpha,1}(z)$, $E_{0,1}(z)=\frac{1}{1-z}$, and $E_{1,1}(z)=e^{z}$.
\end{definition}

\indent The Laplace transform $L(.)$ of two-parameter ML function $E_{\alpha,\beta}(.)$ is defined as
\begin{equation*}
  L\{t^{\beta-1}E_{\alpha,\beta}(vt^{\alpha})\}=\frac{s^{\alpha-\beta}}{s^{\alpha}-v}, \quad t>0,
\end{equation*}
where $t$ and $s$ are the variable in the time domain and Laplace domain, respectively, $v$ is real number, the real part $Re(s)$ of $s$ is $Re(s)>|v|^{\frac{1}{\alpha}}$.
\begin{lemma}\label{FLLE:L1}\cite{21}
Let $V\in C^{1}([t_{0},+\infty),\mathbb{R})$. Then for $0<\alpha<1$,
\begin{equation*}
  {}^{C}_{t_{0}}D^{\alpha}_{t} |V(t)|\leq \operatorname{sgn}(V(t))\,{}^{C}_{t_{0}}D^{\alpha}_{t} V(t).
\end{equation*}
\end{lemma}

\begin{lemma}\cite{22}\label{FLLE:L2}
Let $V:[t_{0},+\infty)\to\mathbb{R}$ be continuous and nonnegative. If there exists a constant $P>0$ such that
\begin{equation}\label{FLLE7}
 {}^{C}_{t_{0}}D^{\alpha}_{t} V(t) \leq -P\,V(t), \qquad 0<\alpha<1,\ t\ge t_{0},
\end{equation}
then
\begin{equation}\label{FLLE8}
V(t)\leq V(t_{0})\,E_{\alpha}\!\big(-P\,(t-t_{0})^{\alpha}\big), \qquad t\ge t_{0}.
\end{equation}
\end{lemma}

\subsection{Existence and uniqueness of the equilibrium point}
In this study, $[x^{*},y^{*}]^{T}$ denotes an equilibrium point of system~\eqref{FLLE1}.
For convenience, we shift the equilibrium to the origin (this causes no loss of generality in the stability
and synchronization analysis).
Define the translated variables
\[
u(t)=x(t)-x^{*},\qquad v(t)=y(t)-y^{*}.
\]
Then system~\eqref{FLLE1} can be rewritten in the equivalent form
\begin{align}\label{FLLE9}
\left\{
\begin{aligned}
{}^{C}_{t_{0}}D^{\alpha}_{t} u(t) &= -A\,u(t)+W\,\hat{f}\big(v(t-\tau(t))\big),\\
{}^{C}_{t_{0}}D^{\alpha}_{t} v(t) &= -B\,v(t)+V\,\hat{g}\big(u(t-\tau(t))\big),
\end{aligned}
\right.
\end{align}
where $\hat{f}(v)=f(v+y^{*})-f(y^{*})$ and $\hat{g}(u)=g(u+x^{*})-g(x^{*})$.
The shifted nonlinearities $\hat{f}(\cdot)$ and $\hat{g}(\cdot)$ inherit the same Lipschitz constants as
$f(\cdot)$ and $g(\cdot)$ (Assumption~1). Therefore, the origin $(u,v)=(0,0)$ is an equilibrium of~\eqref{FLLE9},
and the stability/synchronization analysis can be performed at the origin.
For notational simplicity, we drop the hats and rename $(u,v)$ back to $(x,y)$, obtaining
\begin{align}\label{FLLE10}
\left\{
\begin{aligned}
{}^{C}_{t_{0}}D^{\alpha}_{t} x(t) &= -A\,x(t)+W\,f\big(y(t-\tau(t))\big),\\
{}^{C}_{t_{0}}D^{\alpha}_{t} y(t) &= -B\,y(t)+V\,g\big(x(t-\tau(t))\big).
\end{aligned}
\right.
\end{align}
In the next section, we establish a sufficient criterion guaranteeing global Mittag--Leffler synchronization for the
corresponding drive/response networks.
\section{Main result}

\indent Now, let us establish the sufficient criterion for the GML synchronization concerning the synchronization error system (22).
 \begin{definition}\cite{7}\label{SAY}
The drive system~\eqref{FLLE1} and the response system~\eqref{FLLE4} (under the controls $\mathcal{F}(t)$ and $\mathcal{G}(t)$)
are said to achieve global Mittag--Leffler (GML) synchronization if there exist constants $C>0$ and $P>0$ such that
\begin{equation}\label{FLLE11}
  \|e_{x}(t)\|+\|e_{y}(t)\| \le C\,E_{\alpha}\!\big(-P\,t^{\alpha}\big),\qquad t\ge 0,
\end{equation}
where $e_x(t)=\bar{x}(t)-x(t)$ and $e_y(t)=\bar{y}(t)-y(t)$.
\end{definition}

\indent Now, we will offer the necessary condition of the GML synchronization for the synchronization for the synchronization error system \eqref{FLLE5}.\\
\textbf{\textsl{Finite time synchronization}}
\begin{theorem}\label{GMLSy}
Suppose Assumption~1 holds. Consider the error system \eqref{FLLE5} under the linear error-feedback control \eqref{FLLE6}.
Define
\begin{align}\label{FLLE12}
\eta_{1} &= \min_{1\le i\le n}\big(a_{i}-\phi\big), \qquad
\eta_{2} = \min_{1\le j\le m}\big(b_{j}-\varphi\big), \nonumber\\
\Delta_{1} &= \max_{1\le i\le n}\sum_{j=1}^{m} |w_{ij}|\,L_{f}, \qquad
\Delta_{2} = \max_{1\le j\le m}\sum_{i=1}^{n} |v_{ji}|\,L_{g},
\end{align}
and let
\[
P_{1}=\min\{\eta_{1},\eta_{2}\},\qquad
P_{2}=\max\{\Delta_{1},\Delta_{2}\},\qquad
P=P_{1}-P_{2}.
\]
If $P>0$ (equivalently, $P_{1}>P_{2}$), then the drive system \eqref{FLLE1} and response system \eqref{FLLE4}
achieve global Mittag--Leffler (GML) synchronization in the sense of Definition~\ref{SAY}.

Moreover, for any tolerance $\varepsilon\in(0,V(0))$ the time required to reach $V(t)\le\varepsilon$ can be conservatively estimated as
\begin{equation}\label{FLLE13}
T_{\varepsilon}\le \left[\frac{\Gamma(1+\alpha)}{P}\left(\frac{V(0)}{\varepsilon}-1\right)\right]^{1/\alpha},
\end{equation}
where $V(t)$ is the Lyapunov function defined in \eqref{FLLE14}.
\end{theorem}

\begin{proof}
Consider the Lyapunov function
\begin{equation}\label{FLLE14}
  V(t)=\sum_{i=1}^{n} \|e_{x_i}(t)\| + \sum_{j=1}^{m}\|e_{y_j}(t)\|.
\end{equation}
Using Lemma~\ref{FLLE:L1} and the error dynamics \eqref{FLLE5}, we obtain
\begin{align*}
{}^{C}_{t_{0}}D^{\alpha}_{t} V(t)\leq &
\sum_{i=1}^{n}(-a_{i}+\phi)\|e_{x_i}(t)\|
      +\sum_{j=1}^{m}(-b_{j}+\varphi)\|e_{y_j}(t)\|\\
&\quad +\sum_{i=1}^{n}\sum_{j=1}^{m}|w_{ij}|\,\Big\| f_{j}\big(\bar{y}_{j}(t-\tau(t))\big)-f_{j}\big(y_{j}(t-\tau(t))\big)\Big\|\\
&\quad
      +\sum_{j=1}^{m}\sum_{i=1}^{n}|v_{ji}|\,\Big\| g_{i}\big(\bar{x}_{i}(t-\tau(t))\big)-g_{i}\big(x_{i}(t-\tau(t))\big)\Big\|.
\end{align*}
By Assumption~1 (componentwise Lipschitz continuity), we have
\[
\Big\| f_{j}\big(\bar{y}_{j}(t-\tau(t))\big)-f_{j}\big(y_{j}(t-\tau(t))\big)\Big\|\le L_{f}\|e_{y_j}(t-\tau(t))\|,\]
\[
\Big\| g_{i}\big(\bar{x}_{i}(t-\tau(t))\big)-g_{i}\big(x_{i}(t-\tau(t))\big)\Big\|\le L_{g}\|e_{x_i}(t-\tau(t))\|.
\]
Hence,
\begin{equation}\label{FLLE15}
{}^{C}_{t_{0}}D^{\alpha}_{t} V(t)
\le -P_{1}\,V(t)+P_{2}\,\sup_{s\in[t-\tau_{M},\,t]}V(s),
\end{equation}
The delayed terms are bounded using the Lipschitz property and the Razumikhin condition, which allows replacing delayed states by the supremum of the Lyapunov function over the delay interval.
where $\tau_{M}>0$ is such that $\tau(t)\le\tau_{M}$ for all $t\ge 0$ and $P_{1},P_{2}$ are defined in \eqref{FLLE12}.
Under the standard Razumikhin condition $\sup_{s\in[t-\tau_{M},\,t]}V(s)\le V(t)$, inequality \eqref{FLLE15} implies
\[
{}^{C}_{t_{0}}D^{\alpha}_{t} V(t)\le -P\,V(t), \qquad P=P_{1}-P_{2}>0.
\]
Applying Lemma~\ref{FLLE:L2} yields
\begin{equation}\label{FLLE16}
V(t)\le V(0)\,E_{\alpha}\!\big(-P\,t^{\alpha}\big),\qquad t\ge 0.
\end{equation}
Therefore $V(t)\to 0$ as $t\to\infty$, which proves global Mittag--Leffler synchronization.

Finally, to obtain an explicit conservative time-to-tolerance estimate, we use the standard bound for $0<\alpha<1$ and $x\ge 0$,
\[
E_{\alpha}(-x)\le \frac{1}{1+\dfrac{x}{\Gamma(1+\alpha)}}.
\]
Setting $x=P\,t^{\alpha}$ in \eqref{FLLE16} and solving $V(t)\le\varepsilon$ gives \eqref{FLLE13}.
\end{proof}
\begin{remark}
The condition \(P > 0\) ensures global Mittag--Leffler synchronization and asymptotic contraction of the synchronization error. Nevertheless, this does not exclude the presence of transient local divergence along system trajectories. Such short-term instability is captured by finite-time Lyapunov exponents (FTLE), which may exhibit positive values even when the global synchronization condition holds. Therefore, FTLE provides a complementary local characterization of system dynamics that is not directly reflected in the analytical condition \(P > 0\).
\end{remark}
\begin{remark}
The synchronization results established in Theorem 3.2 guarantee asymptotic convergence of the error system in the sense of Mittag-Leffler stability. This implies that the synchronization error converges to zero as $t \to \infty$.  It is important to note that finite-time convergence is not achieved in the present framework. The Mittag-Leffler decay provides a generalized form of exponential stability, which is standard for fractional-order systems.
\end{remark}
\begin{remark}
The synchronization condition established in Theorem 3.2 depends explicitly on key system parameters, including the fractional order $\alpha$, the time-varying delay $\tau(t)$, and the controller gains $\phi$ and $\varphi$. Variations in these parameters influence both the convergence rate and transient dynamics of the system. In particular, smaller values of $\alpha$ introduce stronger memory effects, resulting in slower Mittag--Leffler convergence and prolonged transient behavior. Larger delays increase the contribution of the delayed term in the Razumikhin condition, effectively reducing the stability margin $P = P_1 - P_2$, which may slow down synchronization. On the other hand, appropriately chosen controller gains enhance the damping effect and improve convergence speed by increasing $P_1$. These parameter variations are also reflected in the finite-time Lyapunov exponent (FTLE) profiles, where changes in $\alpha$, $\tau(t)$, and control gains affect the magnitude and duration of transient instability.  Therefore, while the condition $P > 0$ guarantees global Mittag--Leffler synchronization, the transient performance and local stability characteristics remain sensitive to system parameters.
\end{remark}

  \begin{remark}
  Earlier studies such as \cite{21} realized finite-time convergence in fractional-order neural networks by employing feedback controllers whose gains explicitly depend on the system's delay. Such approaches require accurate knowledge of the delay value. In many real applications, however, obtaining precise delay measurements is challenging and sometimes not feasible at all. To avoid this limitation, the present work develops synchronization controllers for FCVBAMNNs that do not rely on delay-dependent terms, thereby eliminating the need for delay measurement altogether.
  \end{remark}
  \section{Numerical Simulation}
  To demonstrate the effectiveness of the synchronization criterion established in Theorem \ref{GMLSy} and the machine learning validation framework of Theorem  \ref{GMLSy}, a numerical example of the FOCVBAMNNs is presented.\\
  To ensure numerical transparency, reproducibility, and practical relevance, this section presents the numerical diagnostic workflow used to assess local instability and short-horizon predictability in complex-valued fractional-order BAM neural networks. The workflow integrates micro-ensemble-based FTLE estimation with kNN prediction-error analysis, forming a coherent finite-data diagnostic framework.  All simulation results and figures presented in this paper are generated using codes developed by the authors based on the proposed model and methodology.\\
  \textbf{\textsl{Example:}}
  Given the following FCVBAMNNs as the drive system:

\begin{align}\label{FLLE21}
{}^{C}D_{0}^{\alpha} x_{i}(t) &= -a_{i} x_{i}(t)
+ \sum_{j=1}^{m} w_{ij} f_{j}\big(y_{j}(t - \tau(t))\big)
+ M_{i}, \nonumber \\
{}^{C}D_{0}^{\alpha} y_{j}(t) &= -b_{j} y_{j}(t)
+ \sum_{i=1}^{n} v_{ji} g_{i}\big(x_{i}(t - \tau(t))\big)
+ N_{j}.
\end{align}
Where $\alpha =0.9$, $\tau=0.1$, represent the time varying delay, the activation function are consider as $f_{j}(y_{j}(t))=0.8\tanh(z)$, $g_{i}(x_{i}(t))=0.9\tanh(z)$, for $i=j=1,2$. As a result, the drive system \eqref{FLLE21} is simulated with the initial conditions
$x(0) = [0.1+0.05i,\; 0.2-0.1i]^{T}$ and
$y(0) = [-0.1+0.02i,\; 0.1-0.05i]^{T}$. The resulting response system is:
\begin{align}\label{FLLE22}
{}^{C}D_{0}^{\alpha} \bar{x}_{i}(t) &= -a_{i} \bar{x}_{i}(t)
+ \sum_{j=1}^{m} w_{ij} f_{j}\big(\bar{y}_{j}(t - \tau(t))\big)
+ M_{i}+\mathcal{F} (t),\nonumber \\
{}^{C}D_{0}^{\alpha} \bar{y}_{j}(t) &= -b_{j} \bar{y}_{j}(t)
+ \sum_{i=1}^{n} v_{ji} g_{i}\big(\bar{x}_{i}(t - \tau(t))\big)
+ N_{j} +\mathcal{G}(t).
\end{align}
The definition of controller is proved for
\begin{align*}
 \mathcal{F}=&\phi e_{x}(t)\nonumber\\
 \mathcal{G}=&\varphi e_{y}(t)
\end{align*}
Based on the fractional drive system \eqref{FLLE21} and the fractional response system \eqref{FLLE22}, the corresponding fractional error (discrepancy) system is obtained as follows.
\begin{align} \label{FLLE23}
{}^{C}D_{0}^{\alpha} e_{x_{i}}(t)
&= -a_{i}\, e_{x_{i}}(t)
  + \sum_{j=1}^{m} w_{ij}\, f_{j}\big(e_{y_{j}}(t - \tau(t))\big)
  + \mathcal{F}(t), \nonumber \\[6pt]
{}^{C}D_{0}^{\alpha} e_{y_{j}}(t)
&= -b_{j}\, e_{y_{j}}(t)
  + \sum_{i=1}^{n} v_{ji}\, g_{i}\big(e_{x_{i}}(t - \tau(t))\big)
  + \mathcal{G}(t).
\end{align}
The parameters are specified as follows:
\begin{equation*}
a_{i}=  \begin{bmatrix}
    3 & 0\\
    0 & 2.5
  \end{bmatrix}, b_{j}=\begin{bmatrix}
                         2.8 & 0 \\
                         0 & 3.2
                       \end{bmatrix}, w_{ij}=\begin{bmatrix}
                                               0.2+0.1i & -0.05+0.02i \\
                                               0.15-0.03i & 0.04+0.08i
                                             \end{bmatrix},
                                             \end{equation*}
                                             \begin{equation*}
                                             v_{ji}=\begin{bmatrix}
                                                                     0.1+0.05i & -0.2+0.03i \\
                                                                     -0.08i & 0.09+0.02i
                                                                   \end{bmatrix}
\end{equation*}

Let the controller be $\phi=0.5$, $\varphi=0.4$.  Upon completing the analysis in Theorem \ref{GMLSy}, the result is $ P_{2}=\max\{ 0.4159, 0.3669\}=0.4159>0$, $P_{1}=\min \{2.0, 2.4 \}=2.0>0$  with $P_{1}>P_{2}>0$. Thus, condition \eqref{FLLE22} holds. By Theorem \ref{GMLSy}, the system described in \eqref{FLLE23} realizes GML synchronization. These findings demonstrate that the selected control parameters successfully govern the system behavior, guaranteeing synchronization under the specified stability criteria. The settling time is evaluate as $T_{1}^{*}=1.257$.\\
\indent  This fig. \ref{fig:lle-a} shows that the  compares the real component of the first state between the master and response systems.
Without control, the trajectories diverge, while the controller forces rapid convergence. The behavior confirms effective synchronization in the real domain. This fig. \ref{fig:lle-b} shows that the  imaginary portion of the first state exhibits mismatch in the uncontrolled case.
Applying the controller ensures the response system tracks the master accurately. The convergence demonstrates successful synchronization of the imaginary component. In fig. \ref{fig:lle-c} the second states real part shows significant deviation when no control is applied.
With the controller, the trajectories align closely and converge smoothly. This validates the controller's ability to regulate higher order states.
Fig.\ref{fig:lle-d} uncontrolled dynamics lead to discrepancies in the imaginary portion of the second state. The controlled response precisely matches the drive system trajectory. This verifies consistent synchronization across complex-valued components.  Fig. \ref{fig:lle-e} the error remains non-zero throughout the simulation in the absence of control. This indicates that the master and response systems cannot synchronize naturally. The plot highlights the necessity of feedback control for synchronization. Fig. \ref{fig:lle-f} with control applied, the synchronization error rapidly converges to zero. This demonstrates successful tracking and error suppression. The result confirms the efficacy of the proposed control law. In fig. \ref{fig:lle-g} the Lyapunov function decreases monotonically, consistent with theoretical predictions. The curve stays within the Mittag-Leffler bound, verifying fractional-order stability. This confirms the satisfaction of the conditions in Theorem \ref{GMLSy}. The fig. \ref{fig:lle-h} shows the temporal evolution of neuron outputs in the fractional-order BAM network. Fluctuations reflect nonlinear and fractional-order dynamics of the system. The trajectory showcases the complex-valued behavior under the chosen parameters.

\begin{figure*}[!htbp]
\centering

\subfloat[Real part of $x(t)$ for the drive and response systems\label{fig:lle-a}]{
  \includegraphics[width=0.45\textwidth]{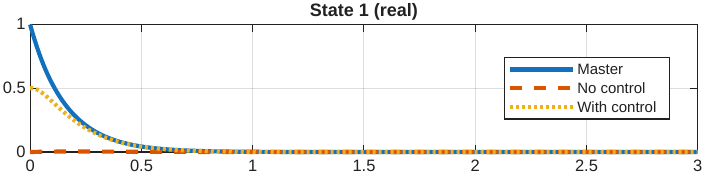}

}\hfill
\subfloat[ Imaginary part of $x(t)$   \label{fig:lle-b}]{
  \includegraphics[width=0.45\textwidth]{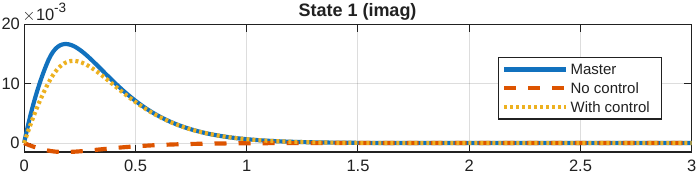}

}

\vspace{3mm}

\subfloat[ Real part of $y(t)$ \label{fig:lle-c}]{
  \includegraphics[width=0.45\textwidth]{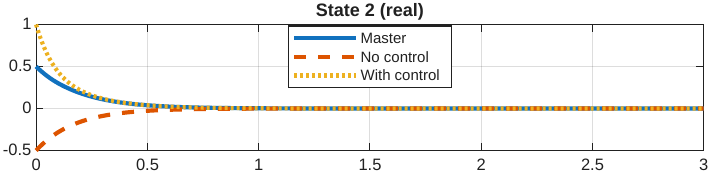}

}\hfill
\subfloat[ Imaginary part of $y(t)$  \label{fig:lle-d}]{
  \includegraphics[width=0.45\textwidth]{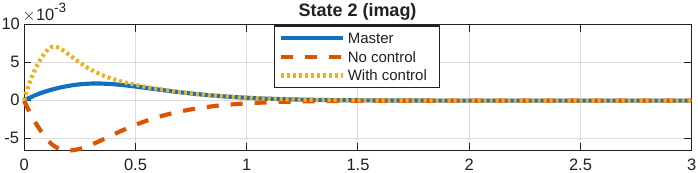}

}

\vspace{3mm}

\subfloat[Synchronization error without the designed controller  \label{fig:lle-e}]{
  \includegraphics[width=0.45\textwidth]{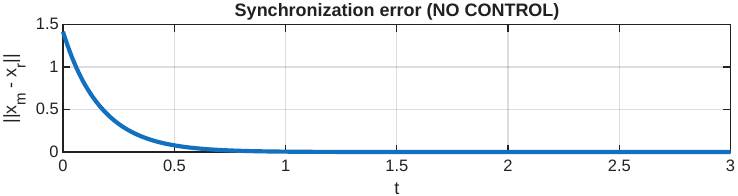}

}\hfill
\subfloat[Synchronization error under the designed controller   \label{fig:lle-f}]{
  \includegraphics[width=0.45\textwidth]{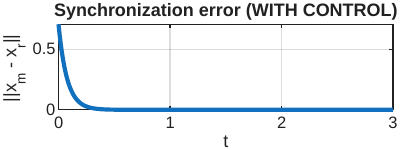}

}

\vspace{3mm}

\subfloat[ Lyapunov function evolution and the Mittag-Leffler stability bound \label{fig:lle-g}]{
  \includegraphics[width=0.45\textwidth]{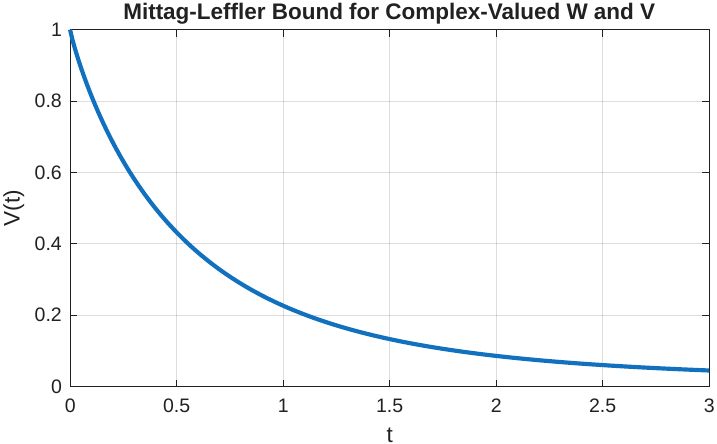}

}\hfill
\subfloat[ Real part of the complex-valued BAM neural network time series   \label{fig:lle-h}]{
  \includegraphics[width=0.45\textwidth]{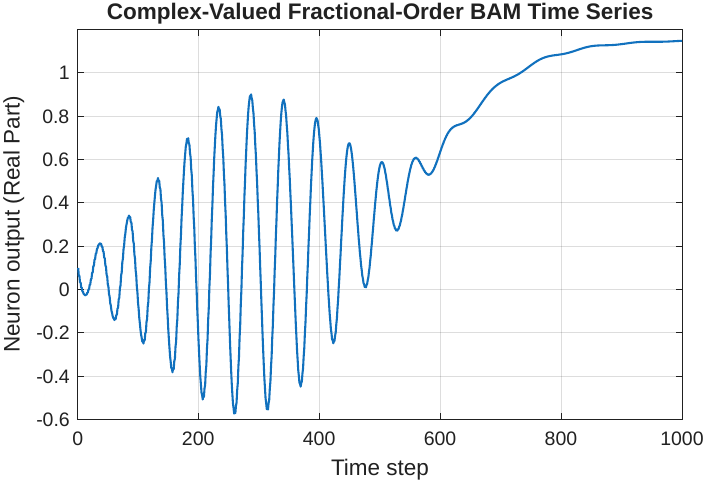}

}

\caption{State trajectories of systems (1), (4), and (6) with and without the designed controller.}
\label{fig:lle-all}

\end{figure*}
\section{Data-Driven Local Instability Analysis: Micro-Ensembles and kNN Proxies}
This combined approach enables a multi-scale characterization of system dynamics, where global synchronization is ensured analytically, and local instability patterns are captured empirically.\\
The analytical condition \(P > 0\) established in Theorem~3.2 guarantees global Mittag--Leffler synchronization, which implies asymptotic decay of the synchronization error. However, this sufficient condition does not explicitly characterize transient or local instability along system trajectories. To address this limitation, the finite-time Lyapunov exponent (FTLE) and kNN-based prediction-error proxies are introduced as data-driven diagnostics. These quantities capture short-horizon divergence and local predictability, thereby providing a trajectory-level complement to the global stability guarantee ensured by \(P > 0\).\\
\indent In nonlinear fractional-order systems, especially those with complex-valued states, computing Lyapunov exponents through variational equations is often impractical because the Jacobian is difficult to evaluate accurately along the trajectory. To address this, the micro-ensemble technique generates several nearby trajectories by applying small perturbations to the reference state. By tracking the short-term growth of distances within this ensemble, we obtain a purely data-driven estimate of the local finite-time Lyapunov exponent (FTLE) without requiring any analytical linearization. This approach provides a high-resolution characterization of local instability and allows us to quantify synchronization behavior more reliably within the fractional-order BAM neural network.
\subsection{Local FTLE estimation from micro-ensembles}
Local FTLE estimation from micro-ensembles offers a fast, data-driven way to measure short-term divergence in complex dynamical systems. By evolving a small set of locally perturbed states, it captures transient sensitivity without requiring full variational equations. This makes it well suited for analyzing nonstationary behavior and localized stability changes.\\
We assume the availability of:
\begin{itemize}
  \item 	A main trajectory  $\{ x_{main} (t_{i})\}_{i=0}^{T-1}$, sampled at times $t_{i}$ with an approximately constant time step  $\Delta t=median (t_{(i+1)}-t_{i})$, where $x_{main}(t_{i})\in \mathbb{R}^{8}$ collects the real and imaginary parts of the four complex BAM variables;
  \item 	A set of micro-ensembles indexed by a reference identifier RefID. For each reference time $t_{ref}$, we have $N$ micro-runs (labelled by Micro $ID=1,\dots,N$) of length $H+1$ steps, starting from small perturbations of the same reference state and recorded at discrete steps $h=0,\dots,H$ with times $t_{ref}+h\Delta t$.
\end{itemize}
	Only reference points that admit at least $N_{min}$ complete micro-runs and at least $H_{min}+1$ aligned steps are retained for analysis.
\subsubsection{Alignment and distance growth}
For each reference time $t_r$:
\begin{enumerate}
    \item Identify the subset of main-trajectory indices $\{i(h)\}_{h=0}^{H}$ such that $
t_{i(h)} = t_{\mathrm{ref}} + h\Delta t.$  This defines the base trajectory segment $x_{\mathrm{main}}\!\left(t_{i(h)}\right)$ used for comparison.

    \item For each micro-run $k = 1,\ldots,N$ and step $h = 0,\ldots,H$, extract the corresponding micro state $ x^{(k)}(h) \in \mathbb{R}^8 .$

    \item Compute the Euclidean distance in state space between each micro-run and the main trajectory at the same time: $
d_k(h) = \left\| x^{(k)}(h) - x_{\mathrm{main}}\!\left(t_{i(h)}\right) \right\|_2,$\\

 $k = 1,\ldots,N,\; h = 0,\ldots,H.$

To avoid numerical issues at very small separations, each $d_k(h)$ is bounded below by a small constant $d_{\min} > 0$, i.e.
$d_k(h) \leftarrow \max\!\left(d_k(h),\, d_{\min}\right).$

    \item 	Define the logarithmic distance for each micro-run and step as $\log d_k(h)$, and average over the ensemble to obtain the logarithm of the geometric mean distance:
\[
L(h) = \langle \log d(h) \rangle
     = \frac{1}{N} \sum_{k=1}^{N} \log d_k(h),
     \qquad h = 0,\ldots,H.
\]

Equivalently, the geometric mean distance is
\[
\overline{d}(h) = \exp\!\big(L(h)\big).
\]

We associate to each step an elapsed time
\[
\tau_h = h\,\Delta t .
\]
\end{enumerate}
	Equivalently, the geometric mean distance is $\bar{d}(h)=exp(L(h))$. We associate to each step an elapsed time $\tau(j)=j\Delta t$.

\subsubsection{One-line vs two-line fit.}
\label{sec:mae_fit}
The local FTLE is extracted from the short-horizon growth of $L(h)$ as a function of time $\tau_h$. Since the behaviour at large horizons may be affected by saturation, nonlinearity or noise, we exclude the initial step $h=0$ from the fit and consider only $h \geq 1$. We then fit either:

\begin{itemize}
    \item A single-line model over all available horizons:

\[
    L(h) \approx a_{\text{all}} \, \tau_h + b_{\text{all}},
    \]

    \item Or a two-line (piecewise linear) model with a breakpoint $\tau^*$:
\[
    L(h) \approx
    \begin{cases}
        a_L \, \tau_h + b_L, & \tau_h \leq \tau^*, \\
        a_R \, \tau_h + b_R, & \tau_h \geq \tau^*,
    \end{cases}
    \]

    with at least two data points in each segment.
\end{itemize}

In both cases, the slopes and intercepts are obtained by ordinary least-squares regression. For the two-line model, we scan all admissible split indices and select the break point $\tau^*$ that minimizes the sum of mean absolute errors (MAE) of the left and right segments. We then compare:

\[
\text{MAE}_{\text{single}} \quad \text{(MAE of the single-line fit)},
\]

\[
\text{MAE}_{\text{two}} = \tfrac{1}{2} \left( \text{MAE}_L + \text{MAE}_R \right) \quad \text{(average MAE of the best two-line fit)}.
\]

The two-line model is accepted if it yields at least a 10\% reduction in MAE, i.e.

\[
\frac{\text{MAE}_{\text{single}} - \text{MAE}_{\text{two}}}{\text{MAE}_{\text{single}}} \geq 0.10.
\]
Otherwise, we fall back to the single-line model.
The 10 \% MAE reduction threshold is adopted as a practical criterion to balance goodness-of-fit and model complexity in the presence of limited and noisy data. This heuristic is consistent with standard model selection principles, where additional model complexity is justified only if it yields a significant improvement in fit quality. We note that alternative criteria such as the Bayesian Information Criterion (BIC) or cross-validation lead to qualitatively similar model selections in our experiments.
\subsubsection{Local Lyapunov proxy}
For each reference time $t_{\text{ref}}$ we define the local FTLE estimate
$\hat{\lambda}_1(t_{\text{ref}})$ as the short-horizon slope:

\[
\hat{\lambda}_1(t_{\mathrm{ref}})=
\begin{cases}
a_{\mathrm{all}}, & \text{if the single-line model is selected},\\[6pt]
a_{L},            & \text{if the two-line model is selected}.
\end{cases}
\]

Thus, $\hat{\lambda}_1(t_{\text{ref}})$ measures the initial exponential growth
(or contraction) rate of the geometric-mean perturbation norm over the ensemble
of micro-runs starting from the reference state. Collecting these values over
all reference points yields a time series $\hat{\lambda}_1(t_{\text{ref}})$ that
characterises the local finite-time Lyapunov behaviour along the main trajectory.

\subsection{Results for Method I: Micro-ensemble FTLE}
The micro-ensemble is generated by applying small perturbations to the reference state. Specifically, for each reference point $x(t_{\text{ref}})$, the initial conditions of the micro-runs are constructed as
\[
x^{(k)}(0) = x(t_{\mathrm{ref}}) + \delta^{(k)},
\]
where $\delta^{(k)}$ are small random perturbation vectors sampled from an isotropic distribution in $\mathbb{R}^d$. The perturbations are normalized to have approximately uniform magnitude, and their norm is chosen to be of order $10^{-3}$ to ensure that the dynamics remain in the local linear regime.

\textbf{\textsl{Method-I: Local FTLE estimation from micro-ensembles}}\\
For the BAM test case, the micro-ensemble FTLE estimator is applied using a simulation time step of
$\Delta t = 0.05$, with $n_{\text{ref}} = 200$ reference states. Each reference point includes at least
$N_{\min} = 5$ micro-runs, and a minimum of $H_{\min} = 3$ forward steps is used for the short horizon
fitting procedure. The number of reference points $n_{\text{ref}}=200$ is selected to provide sufficient coverage of the trajectory for capturing temporal variations in local stability. Empirical tests indicate that increasing this number does not significantly alter the qualitative FTLE patterns, while smaller values reduce resolution.\\

For every reference time $t_{\text{ref}}$, an ensemble of micro-trajectories is generated, and the
geometric mean deviation $\bar{d}(\tau)$ from the main trajectory is computed as a function of elapsed
time $\tau = h\Delta t$. The logarithmic profile is then formed as
$L(h) = \langle \log d(h) \rangle$, which is fitted using either a one line or two line MAE-optimal
linear model for $h \ge 1$. The slope of the selected short horizon segment provides the local FTLE
$\hat{\lambda}_{1}(t_{\text{ref}})$.
For each reference point, multiple micro-runs are generated. Only reference points with at least $N_{\min}=5$ valid micro-trajectories are retained to ensure statistical robustness of the geometric mean distance used in FTLE estimation.

\begin{table}[h!]
\centering
\caption{Micro-ensemble data used for FTLE estimation at two reference times.
Each table lists the short-time offset $\tau$, the geometric mean distance
$d_{\text{geo}}(\tau)$, and its logarithm $\log(d_{\mathrm{geo}}(\tau))$.}
\label{tab:micro ensemble data}
\begin{tabular}{c|ccc|ccc}
\hline
\multirow{2}{*}{Step}
& \multicolumn{3}{c|}{Dataset A ($t_{\text{ref}} \approx 5.65$)}
& \multicolumn{3}{c}{Dataset B ($t_{\text{ref}} \approx 26.75$)} \\ \cline{2-7}
& $\tau$ & $d_{\text{geo}}$ & $\log(d_{\text{geo}})$
& $\tau$ & $d_{\text{geo}}$ & $\log(d_{\text{geo}})$ \\ \hline
0  & 0.00 & 0.052576 & -2.94549  & 0.00 & 0.053063 & -2.93627 \\
1  & 0.05 & 0.077960 & -2.55156  & 0.05 & 0.097035 & -2.33269 \\
2  & 0.10 & 0.147455 & -1.91423  & 0.10 & 0.156026 & -1.85773 \\
3  & 0.15 & 0.263993 & -1.33183  & 0.15 & 0.239766 & -1.42809 \\
4  & 0.20 & 0.417097 & -0.87444  & 0.20 & 0.372642 & -0.98714 \\
5  & 0.25 & 0.572560 & -0.55764  & 0.25 & 0.535969 & -0.62368 \\
6  & 0.30 & 0.702181 & -0.35356  & 0.30 & 0.688700 & -0.37295 \\
7  & 0.35 & 0.809952 & -0.21078  & 0.35 & 0.809440 & -0.21141 \\
8  & 0.40 & 0.929226 & -0.07340  & 0.40 & 0.912190 & -0.09191 \\
9  & 0.45 & 1.093881 &  0.08973  & 0.45 & 1.031822 &  0.03133 \\
10 & 0.50 & 1.305382 &  0.26650  & 0.50 & 1.200184 &  0.18247 \\
11 & 0.55 & 1.528704 &  0.42442  & 0.55 & 1.421782 &  0.35191 \\
12 & 0.60 & 1.719864 &  0.54225  & 0.60 & 1.662246 &  0.50817 \\
\hline
\end{tabular}
\end{table}
In this table \ref{tab:micro ensemble data} these data show how the geometric mean distance between micro-trajectories and the main trajectory grows over the short horizon. The near-linear rise in  $\log d_{\text{geo}} (\tau)$ provides the finite-time Lyapunov exponent at that reference point.

Figure~\ref{fig:ftle-growth}, presents two representative growth curves for
$\text{RefID} = 2$ and $\text{RefID} = 32$. The blue markers indicate the empirical values of $L(h)$,
while the solid lines correspond to the MAE-optimal linear fits chosen through the one vs.~two line
selection process. The red marker at $\tau = 0$ denotes the initial perturbation magnitude; although
shown for completeness, it is excluded from the regression because the FTLE quantifies the divergence
relative to the initial deviation. The slope of the fitted short horizon segment is interpreted as the
local FTLE at each reference state. \\
The observed oscillations in FTLE values indicate alternating phases of local expansion and contraction. In particular, intervals of relatively high FTLE correspond to transient sensitivity to perturbations, although the global synchronization condition $P>0$ ensures that such effects do not persist asymptotically. This highlights the distinction between local instability and global stability in fractional-order systems.
\begin{figure}[H]
\centering

\subfloat[FTLE growth examples for ref002.]{
  \includegraphics[width=0.45\textwidth]{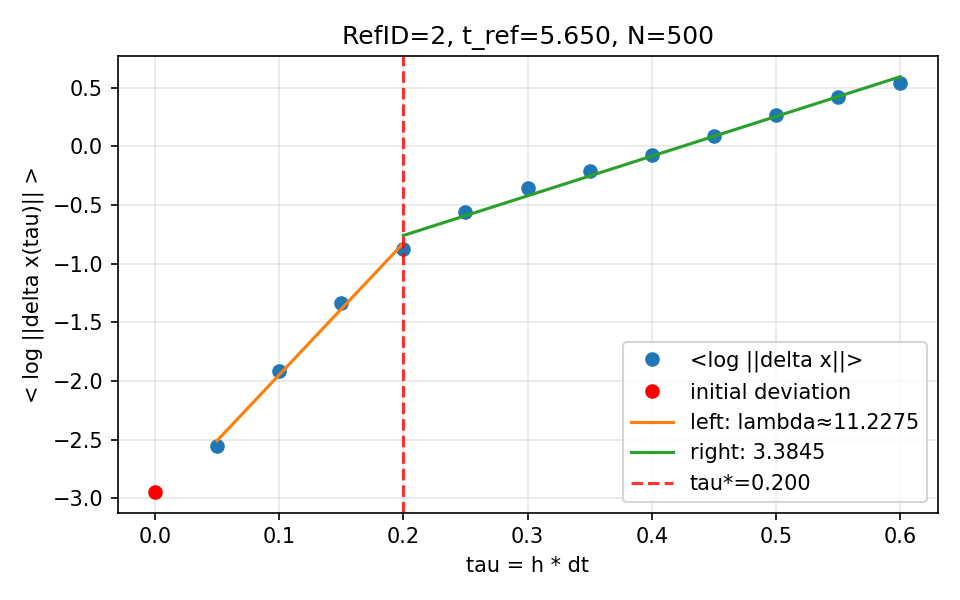}
  \label{fig:figure9}
}
\hfill
\subfloat[FTLE growth examples for ref032.]{
  \includegraphics[width=0.45\textwidth]{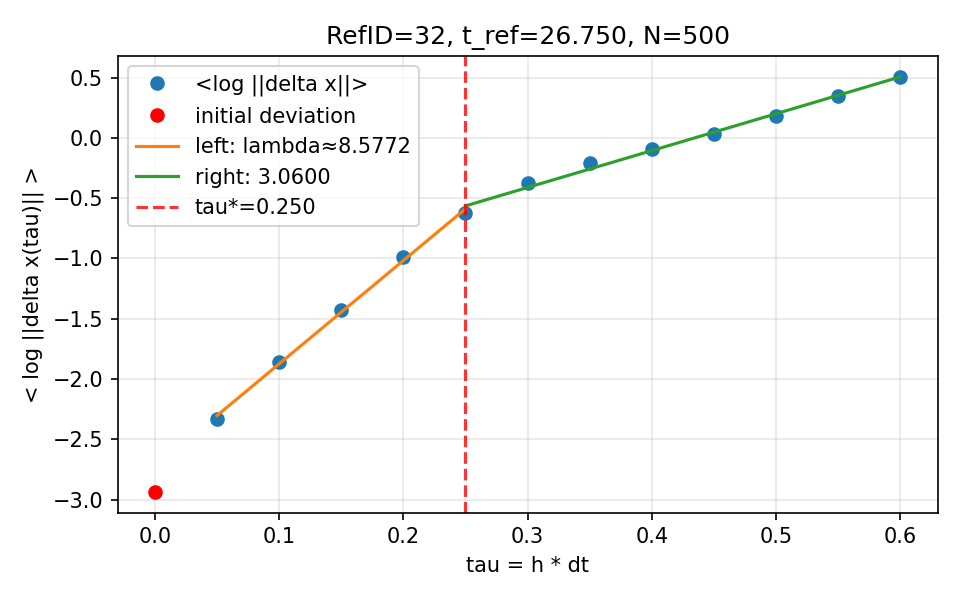}
  \label{fig:figure10}
}

\caption{Examples of ensemble-averaged growth curves
$L(h)=\langle \log \|\delta x(h)\| \rangle$ and their MAE--optimal linear fits.
Blue markers show the empirical profile as a function of $\tau = h\Delta t$, while
the solid lines indicate the selected short-horizon fit (one- or two-line model).
The red marker at $\tau = 0$ denotes the initial ensemble deviation, plotted for
completeness but excluded from the slope estimation.}
\label{fig:ftle-growth}

\end{figure}

Aggregating the local estimates over all reference times yields the FTLE time series
$\hat{\lambda}_{1}(t_{\text{ref}})$ shown in Figure~\ref{fig:figure11}.
The values oscillate between strongly contracting and strongly expanding regimes
(approximately between $4.5$ and $16$) in a nearly periodic, almost sinusoidal
pattern over the full time interval. No significant long term trend is observed;
both the amplitude and shape of the oscillations remain stable, indicating a
persistent quasi periodic modulation of local instability along the main trajectory.
\begin{figure}[htbp]
  \centering
  \includegraphics[width=\textwidth]{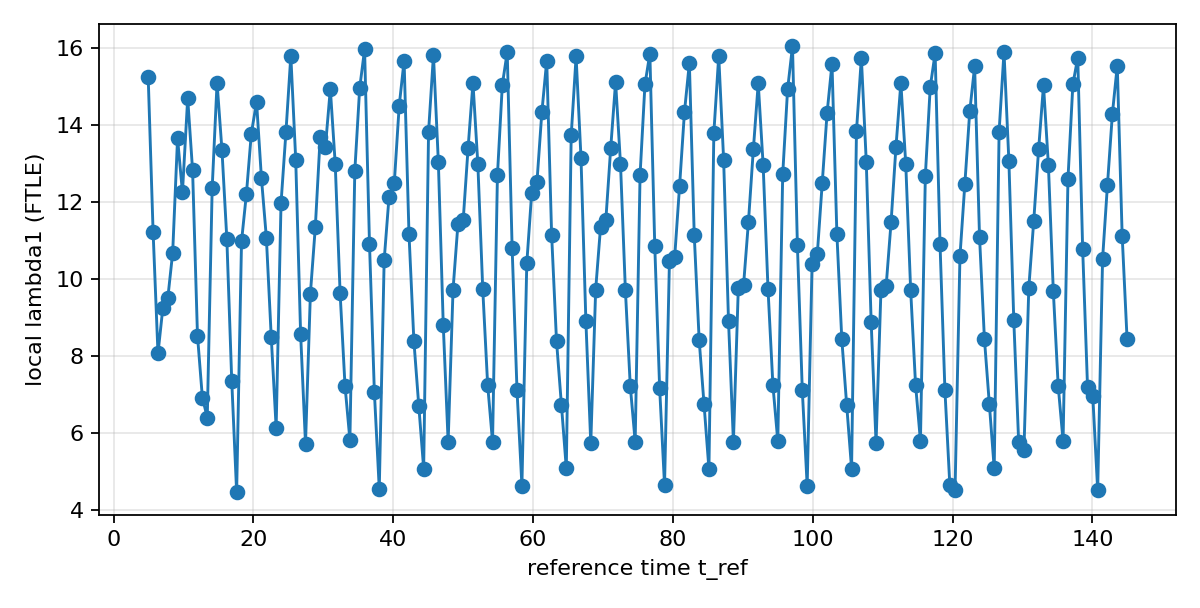}
  \caption{Local FTLE $\hat{\lambda}_{1}(t_{ref})$ estimated from micro-ensembles along the main BAM trajectory. Each marker corresponds to a reference state; the curve reveals a pronounced, nearly sinusoidal modulation of the local growth rate with no appreciable long--term drift.}
  \label{fig:figure11}
\end{figure}
\subsection{Method II: kNN prediction--error Lyapunov proxies}
\noindent\textbf{Rationale and relation to Method I.} The kNN proxy estimates local instability from short-horizon prediction errors rather than from explicit perturbation growth. We use the same MAE-based one-line vs.~two-line fitting and short-horizon slope extraction described in Section~\ref{sec:mae_fit} so that the resulting slopes are directly comparable across the micro-ensemble FTLE and the kNN proxies.

\noindent\textbf{Overview.} Building on the prediction--error Lyapunov proxy introduced in \cite{23}, we implement two kNN-based variants:
\begin{enumerate}[leftmargin=*,itemsep=2pt]
  \item \textbf{Method II-A (Full-state kNN):} predict the full real representation of the complex BAM state, $x(t)\in\mathbb{R}^{8}$ (real and imaginary parts stacked).
  \item \textbf{Method II-B (Modulus kNN):} predict the scalar modulus $r(t)=\|x(t)\|_{2}$, which reduces the output dimension and typically yields smoother error-growth curves.
\end{enumerate}
\noindent In what follows we present the full-state construction first and then the modulus-based modification.

In addition to the FTLE-based construction of the local Lyapunov proxy, we implemented a kNN-based variant that uses short-term prediction errors within each micro-ensemble. The goal is to estimate the local instability around a reference state $x(\text{ref})$ without explicit reference to the main trajectory.
\subsubsection{Micro-ensemble and delay embedding}
For each reference identifier RefID we collect all micro-runs starting from small perturbations of the same base state $x(t_{\text{ref}})$. After truncation to a common length $L_\text{micro}$, each micro-run provides a sequence
We consider
\[
x^{(k)}(j) \in \mathbb{R}^8, \qquad j = 0, \ldots, L_{\text{micro}} - 1,
\]

where $k$ indexes the micro-run and the 8 components correspond to the real and
imaginary parts of the four BAM variables. We fix a delay-embedding length
$m = \texttt{LOOK\_BACK}$ and define the maximal prediction horizon $H_{\max}$ so that

\[
H_{\max} = L_{\text{micro}} - m,
\]

ensuring that all horizons share the same number of usable points.

For each micro-run we construct a single feature vector by stacking the first $m$ states,

\[
u^{(k)} = \big(x^{(k)}(0), \, x^{(k)}(1), \, \ldots, \, x^{(k)}(m-1)\big) \in \mathbb{R}^{8m},
\]

and, for each prediction horizon $h$ in the range $h_{\min} \leq h \leq H_{\max}$,
define the corresponding target state

\[
y_h^{(k)} = x^{(k)}(m+h-1) \in \mathbb{R}^8.
\]

Thus each micro-run contributes one feature vector and a set of future states at all horizons.
\subsubsection{Horizon-wise kNN prediction.}
The kNN method is adopted due to its simplicity and nonparametric nature, making it suitable for small-sample micro-ensemble data. It effectively captures local geometric structure in the state space and provides a direct measure of short-horizon predictability without requiring model training or parameter tuning.\\
Let $K$ be the number of micro-runs available for a given reference time. We split the $K$ samples into a training set and a test set, using a user-specified fraction $\text{test\_size}$ for testing. The effective number of neighbours used by the kNN regressor is

\[
\min(\text{k\_neighbors}, K_{\text{train}}),
\]

so that the method remains well-defined even for small ensembles.

For each horizon $h$ we train an independent kNN model

\[
f_h : \mathbb{R}^{8m} \to \mathbb{R}^8
\]

on the training pairs $\big(u^{(k)}, y_h^{(k)}\big)$. Distance-weighted kNN with Euclidean metric is used throughout.

On the test set we evaluate the prediction error in $\mathbb{R}^8$,

\[
e_h^{(k)} = \left\| y_h^{(k)} - f_h\big(u^{(k)}\big) \right\|_2,
\]

and guard against numerical underflow by replacing any zero error with a small constant

\[
\varepsilon = \text{MIN\_ERROR}.
\]
\subsubsection{Prediction--error growth curve}
For each horizon $h$ we aggregate the test errors by their geometric mean,
We define

\[
G(h) = \exp\!\left( \frac{1}{N_{\text{test}}} \sum_{k \in J} \log e_h^{(k)} \right),
\]

where $J$ is the test set and $N_{\text{test}} = |J|$. We then define

\[
y(h) = \log G(h), \qquad \tau_h = h \, \Delta t,
\]

with $\Delta t$ the sampling interval of the main trajectory.

The resulting profile $y(h)$ quantifies the short-term growth of kNN prediction
error as a function of physical time $\tau$, and plays the same role as
$\langle \log \| \delta x(\tau) \| \rangle$ in the FTLE-based construction.
\subsubsection{Slope extraction and local Lyapunov proxy}
The sequence $\{(\tau_{h},y(h))\}$ is then passed to the same MAE-based one-line vs two-line fitting procedure used for the FTLE micro-ensembles (see previous subsection). In particular, we fit either a single straight line or a piecewise-linear model with one internal breakpoint, select the model according to the MAE reduction criterion, and take the slope of the left (short-horizon) segment as the local kNN-based Lyapunov proxy $\hat{\lambda}_{1}(t_{\text{ref}})$ for that reference state.
\subsubsection{Results for Method II-A (full-state kNN)}
\textbf{\textsl{Method-II: kNN prediction--error Lyapunov proxy}}\\
We first applied the kNN-based Lyapunov estimator to the full complex BAM state,
using the micro-ensemble data described above. In this experiment, the main time
step was $\Delta t = 0.05$, and we used $n_{\text{ref}} = 200$ reference states.
For each reference, we constructed a delay embedding of length
\text{LOOK\_BACK} = 1 (one past vector state as input), a minimal prediction
horizon $h_{\min} = 1$, a test fraction of \text{test\_size} = 0.3 for the
micro-runs, and a nominal number of neighbours \text{k} = 3 for the kNN predictor.
Only references with at least five micro-trajectories were retained for analysis.

\begin{table}[h!]
\centering
\caption{Prediction-error growth data for the kNN-based Lyapunov proxy. Each dataset contains the short-time offset $\tau$, the
prediction-error magnitude $G(\tau)$, and its logarithm $\log G(\tau)$.}
\label{tab:knn growth data}
\begin{tabular}{c|ccc|ccc}
\hline
\multirow{2}{*}{Step}
& \multicolumn{3}{c|}{Dataset A}
& \multicolumn{3}{c}{Dataset B (ref\_032)} \\ \cline{2-7}
& $\tau$ & $G(\tau)$ & $\log G$
& $\tau$ & $G(\tau)$ & $\log G$ \\ \hline
1  & 0.05 & 0.023292 & -3.75966 & 0.05 & 0.022099 & -3.81222 \\
2  & 0.10 & 0.022823 & -3.77997 & 0.10 & 0.021548 & -3.83748 \\
3  & 0.15 & 0.022535 & -3.79267 & 0.15 & 0.021211 & -3.85323 \\
4  & 0.20 & 0.022406 & -3.79843 & 0.20 & 0.021096 & -3.85867 \\
5  & 0.25 & 0.022350 & -3.80092 & 0.25 & 0.021094 & -3.85875 \\
6  & 0.30 & 0.022311 & -3.80268 & 0.30 & 0.021104 & -3.85830 \\
7  & 0.35 & 0.022374 & -3.79987 & 0.35 & 0.021134 & -3.85687 \\
8  & 0.40 & 0.022659 & -3.78721 & 0.40 & 0.021288 & -3.84959 \\
9  & 0.45 & 0.023246 & -3.76162 & 0.45 & 0.021686 & -3.83108 \\
10 & 0.50 & 0.024122 & -3.72461 & 0.50 & 0.022403 & -3.79854 \\
11 & 0.55 & 0.025177 & -3.68183 & 0.55 & 0.023357 & -3.75687 \\
12 & 0.60 & 0.026315 & -3.63762 & 0.60 & 0.024367 & -3.71454 \\
\hline
\end{tabular}
\end{table}
This table \ref{tab:knn growth data} reports the short-horizon growth of the full-state prediction error $G(\tau)$ produced by the kNN model.
The resulting $\log G(\tau)$ curve reflects how small state-space deviations propagate in the high-dimensional system.\\
Figure~\ref{fig:knn-curves},  shows two representative profiles
$y(h) = \log G(h)$ for $\text{RefID} = 2$ and $\text{RefID} = 32$,
together with the MAE-optimal linear fits. In contrast to the FTLE-based
curves, which typically exhibit a nearly linear growth followed by saturation,
the kNN prediction-error curves display a much less regular behaviour.
The values of $\log G(h)$ fluctuate with horizon and may even decrease for
intermediate $h$, so the overall profile deviates significantly from a simple
exponential trend. The fitted lines therefore capture only a coarse local trend
rather than a clean short--horizon growth regime.
\begin{remark}
The numerical results are consistent with the theoretical Mittag-Leffler decay $V(t) \leq V(0)E_\alpha(-Pt^\alpha)$. A detailed quantitative fitting or residual error analysis could further validate this behavior and will be considered in future work.
\end{remark}

\begin{remark}
 The time-varying delay $\tau(t)$ influences the feasibility of the stability conditions through Razumikhin-type constraints. Larger delays may introduce conservatism, but stability is ensured for admissible bounded delays.
\end{remark}

\begin{figure}[H]
\centering

\subfloat[Shape of the prediction-error curves for ref002.]{
  \includegraphics[width=0.45\textwidth]{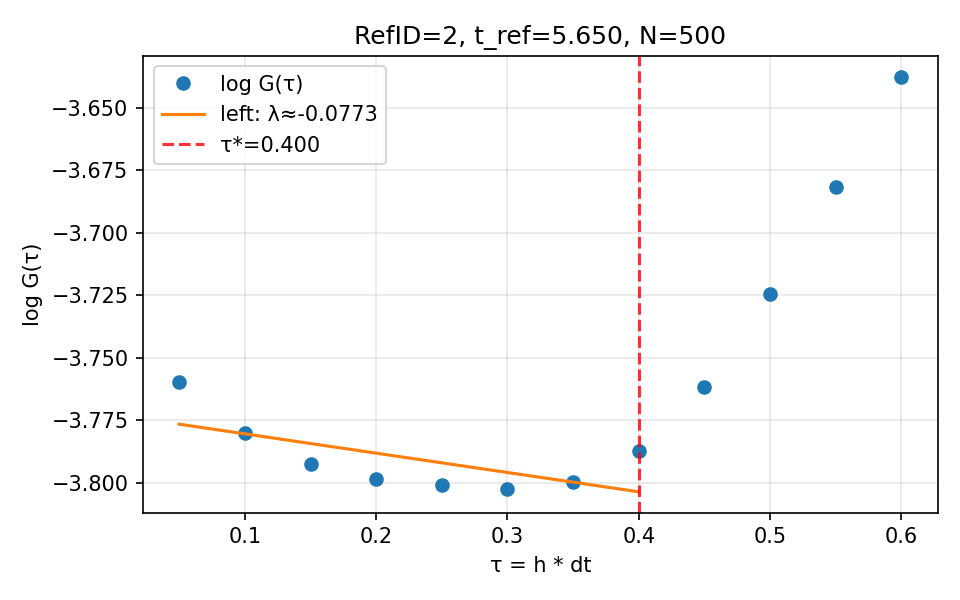}
  \label{fig:figure12}
}
\hfill
\subfloat[Shape of the prediction-error curves for ref032.]{
  \includegraphics[width=0.45\textwidth]{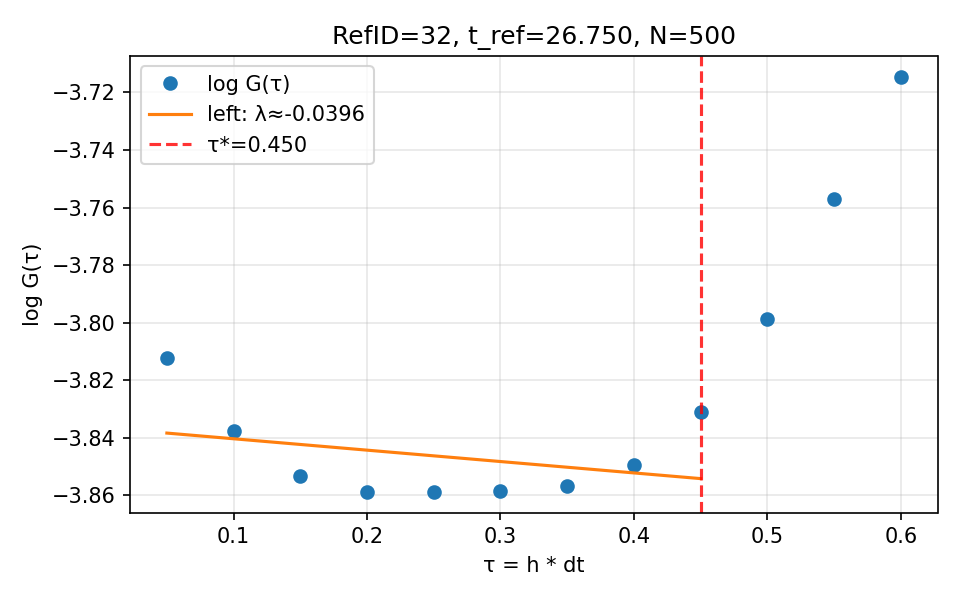}
  \label{fig:figure13}
}

\caption{Examples of kNN-based prediction error curves
$y(h)=\log G(h)$ and their MAE-optimal linear fits for two reference states.
Blue markers show the empirical profile as a function of $\tau=h\Delta t$,
while the solid line indicates the selected short-horizon fit.
The fluctuations reflect the difficulty of predicting all eight state components
from a small micro-ensemble.}
\label{fig:knn-curves}
\end{figure}

This irregularity is consistent with the difficulty of the regression task:
for each horizon, the kNN model must map an 8-dimensional input (or
$8 \times \text{LOOK\_BACK}$ in general) to an 8-dimensional output using
only a small number of micro-runs. With limited training data, the kNN predictor
cannot approximate all eight coordinates with sufficient accuracy, and the
resulting prediction errors are dominated by sampling noise and local nonlinear
effects rather than by a smooth exponential growth.\\
\textbf{\textsl{Time series of kNN-based Lyapunov exponents.}}\\

The slopes extracted from the short horizon fits define the local kNN-based
Lyapunov proxy $\hat{\lambda}_{1}(t_{\text{ref}})$ along the main trajectory
(Figure~\ref{fig:figure14}). In comparison with the FTLE-based series,
the kNN-derived $\hat{\lambda}_{1}$ exhibits a much more irregular pattern:
although a broad oscillatory structure is present, the values fluctuate
substantially from one reference state to the next and do not form a clean
quasi-sinusoidal modulation. This behaviour indicates that the underlying
instability signal is partially obscured by regression noise and finite-sample
variability inherent to the high-dimensional prediction task.

\begin{figure}[htbp]
  \centering
  \includegraphics[width=\textwidth]{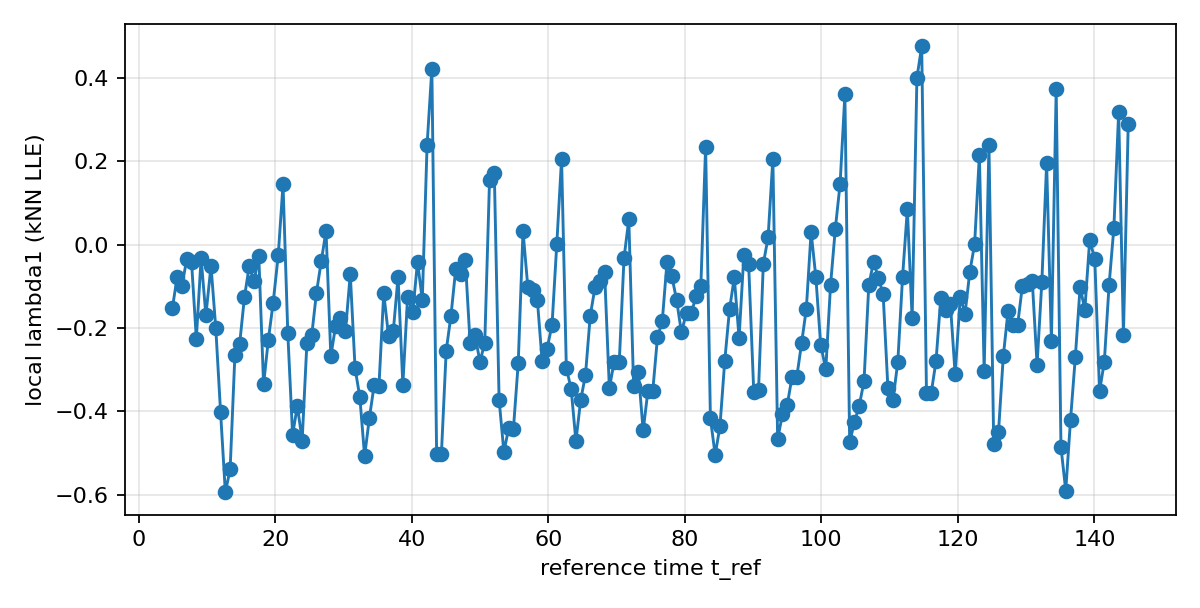}
  \caption{Local Lyapunov proxy $\hat{\lambda}_{1}(t_{\text{ref}})$ obtained from kNN prediction
errors with full-state targets in $\mathbb{R}^{8}$. The time series exhibits sizeable
point-to-point fluctuations and only a weakly visible oscillatory pattern, suggesting
that the high-dimensional regression problem introduces substantial noise into the
slope estimates.}
  \label{fig:figure14}
\end{figure}

\textbf{\textsl{Motivation for scalar-modulus prediction.}}\\
The above results suggest that, with the current ensemble size, the kNN predictor
struggles to provide stable and interpretable error growth for the full
8-dimensional state. To reduce the effective output dimension and improve
robustness, the next subsection switches from predicting all coordinates to
predicting the modulus of the complex BAM variables, and recomputes the
kNN-based Lyapunov proxy using this scalar target.

\subsubsection{Method II-B: Modulus-based kNN proxy}
To reduce the effective output dimension of the kNN regression problem, we consider a scalar variant of the micro--ensemble estimator in which we work with the Euclidean norm (modulus) of the state vector rather than with all components. For each BAM state  $x(t)\in \mathbb{R}^{d} $ (here d=8) we define
\begin{equation*}
  r(t)=\|x(t) \|_{2}
\end{equation*}
and apply the kNN-based prediction--error construction directly to the scalar process $r(t)$.
\subsubsection{Micro-ensemble preprocessing.}
For every reference identifier $\text{RefID}$ with reference time $t_{\text{ref}}$
we collect all micro-runs that start from small perturbations of the same base state
$x(t_{\text{ref}})$. After sorting by the discrete step index, each micro-run provides
a short sequence $ \{ x^{(k)}_l \}_{l=0}^{L_{\text{micro}}-1}.$\\
We then:
\begin{enumerate}
    \item Determine the minimal length $L_{\text{micro}}$ across micro-runs for this reference
    and truncate all runs to this common length.

    \item For each micro-run $k$ compute the modulus time series

\[
    r^{(k)}_l = \| x^{(k)}_l \|_2, \qquad l = 0, \ldots, L_{\text{micro}} - 1,
    \]

    and store it as a 1D array of length $L_{\text{micro}}$.

    \item Discard references with fewer than
    $N_{\min} = \text{MIN\_MICRO\_PER\_REF}$ usable micro-runs or with
    $L_{\text{micro}} \leq \text{LOOK\_BACK} + 1$.
\end{enumerate}
\subsubsection{Delay embedding and horizon-wise regression}
Let $m = \text{LOOK\_BACK}$ be the number of past moduli used as features.
For each reference we set the maximal prediction horizon so that

\[
H_{\max} = L_{\text{micro}} - m,
\]

and choose an integer $h_{\min} \geq 1$.

For each micro-run we construct a single feature vector

\[
u^{(k)} = \big(r^{(k)}_0, r^{(k)}_1, \ldots, r^{(k)}_{m-1}\big) \in \mathbb{R}^m,
\]

and, for each horizon $h \in [h_{\min}, H_{\max}]$, define the scalar target

\[
y_h^{(k)} = r^{(k)}_{m+h-1}.
\]

Thus every micro-run contributes one input--output pair for each prediction horizon.
Across all micro-runs we assemble a design matrix

\[
U \in \mathbb{R}^{K \times m}
\]

and horizon-specific target vectors

\[
y_h \in \mathbb{R}^K,
\]

where $K$ is the number of micro-runs available for the reference.

We split the rows into a training set and a test set using a fixed fraction
$\text{TEST\_SIZE}$ for the test set. For each horizon $h$ we then train a scalar kNN regressor

\[
f_h : \mathbb{R}^m \to \mathbb{R},
\]

on the training pairs $\big(u^{(k)}, y_h^{(k)}\big)$. The regressor uses Euclidean distance,
distance-based weights, and an effective number of neighbours

\[
K_{\text{eff}} = \min(\text{K\_NEIGHBORS}, K_{\text{train}}),
\]

so that $K_{\text{eff}}$ never exceeds the train set size.
\subsubsection{Prediction--error growth on the modulus}
On the test set we evaluate the absolute prediction error for each horizon,

\[
e_h^{(k)} = \left| y_h^{(k)} - f_h\big(u^{(k)}\big) \right|,
\]

and bound it away from zero by a small constant

\[
\varepsilon = \texttt{MIN\_ERROR}
\]

to avoid numerical issues in the logarithm.

The geometric mean of the absolute errors at horizon $h$ is

\[
G(h) = \exp\!\left( \frac{1}{N_{\text{test}}} \sum_{k \in J} \log e_h^{(k)} \right),
\]

where $J$ denotes the set of test micro-runs and $N_{\text{test}} = |J|$.

We then form the logarithmic error profile

\[
l(h) = \log G(h), \qquad \tau_h = h \, \Delta t,
\]

with $\Delta t$ the sampling interval of the main trajectory.

The sequence $\{(\tau_h, l(h))\}$ is analogous to the FTLE-based growth curve
obtained from the ensemble of perturbation norms.
\subsubsection{Slope extraction and Lyapunov proxy}
The curve $l(h)$ as a function of $\tau_h$ is passed to the same MAE-based
one-line vs two-line fitting scheme as in the FTLE micro-ensemble estimator
(see previous subsection). In particular, we either fit a single straight line
over all horizons or a piecewise-linear model with one internal breakpoint,
select the model according to the MAE reduction criterion, and take the slope
of the left (short-horizon) segment as the local modulus-based Lyapunov proxy

\[
\hat{\lambda}_1^{(\text{mod})}(t_{\text{ref}}).
\]

Collecting these values for all reference times yields a time series

\[
\hat{\lambda}_1^{(\text{mod})}(t_{\text{ref}})
\]

that can be compared directly with the FTLE-based and full-state kNN-based
Lyapunov estimates.
\subsubsection{Results for Method II-B (modulus kNN)}
\textbf{\textsl{Method-II : Modulus-based kNN Lyapunov proxy}}\\
We now apply the kNN prediction--error estimator to the scalar modulus
$r(t) = \|x(t)\|_{2}$ of the BAM state, using the configuration described in the
previous subsection. For each reference time $t_{\text{ref}}$, we construct
delay-embedded inputs from the past three values of $r$, use a fraction of $0.3$
of the micro-runs for testing, and train distance-weighted kNN regressors with
up to three neighbours for all prediction horizons. The main time step is
$\Delta t \approx 0.05$, and we consider $n_{\text{ref}} = 200$ reference states.

\textbf{Shape of the modulus error curves.}\\
Figure~\ref{fig:fig15fig16},  shows two representative profiles of the
logarithmic geometric-mean prediction error, $l(h) = \log G(h)$, for
$\text{RefID} = 2$ and $\text{RefID} = 32$. In contrast to the full-state kNN
curves, the modulus-based error profiles exhibit a much more regular structure.
At short horizons, $l(h)$ increases approximately linearly with
$\tau = h \Delta t$, indicating an initial exponential growth of the prediction
error. At larger horizons, the curves bend and gradually saturate, as expected
when the prediction error approaches the typical spread of the modulus within
the micro-ensemble. This saturation behaviour is qualitatively similar to that
observed in the FTLE-based micro--ensemble curves for the full state.

The MAE-based one  vs.\ two line fitting procedure identifies a clear
short-horizon segment in which the linear approximation remains accurate, and
the corresponding slopes provide stable local estimates of the modulus-based
Lyapunov proxy $\hat{\lambda}_{1}^{(\mathrm{mod})}(t_{\text{ref}})$.
\begin{table}[h!]
\centering
\caption{Modulus-based kNN prediction-error data (Method~II-B).
Each dataset lists the short-time offset $\tau$, the modulus prediction-error
$G(\tau)$, and its logarithm $\log G(\tau)$.}
\label{tab:modulus knn data}
\begin{tabular}{c|ccc|ccc}
\hline
\multirow{2}{*}{Step}
& \multicolumn{3}{c|}{Dataset A}
& \multicolumn{3}{c}{Dataset B (ref\_032)} \\ \cline{2-7}
& $\tau$ & $G(\tau)$ & $\log G$
& $\tau$ & $G(\tau)$ & $\log G$ \\ \hline
1  & 0.05 & 0.0009955 & -6.91223 & 0.05 & 0.0011066 & -6.80646 \\
2  & 0.10 & 0.0017557 & -6.34489 & 0.10 & 0.0019784 & -6.22547 \\
3  & 0.15 & 0.0025197 & -5.98361 & 0.15 & 0.0028916 & -5.84594 \\
4  & 0.20 & 0.0028795 & -5.85014 & 0.20 & 0.0033890 & -5.68723 \\
5  & 0.25 & 0.0035150 & -5.65072 & 0.25 & 0.0037889 & -5.57567 \\
6  & 0.30 & 0.0042279 & -5.46606 & 0.30 & 0.0047184 & -5.35628 \\
7  & 0.35 & 0.0053321 & -5.23401 & 0.35 & 0.0055196 & -5.19945 \\
8  & 0.40 & 0.0061914 & -5.08460 & 0.40 & 0.0070465 & -4.95522 \\
9  & 0.45 & 0.0070002 & -4.96182 & 0.45 & 0.0079738 & -4.83159 \\
10 & 0.50 & 0.0066775 & -5.00902 & 0.50 & 0.0074034 & -4.90582 \\
\hline
\end{tabular}
\end{table}

These values show  in the table \ref{tab:modulus knn data} the prediction-error growth in the scalar modulus $\| x(t)\|$, yielding smoother and more stable curves than the full-state method. The monotonic increase in $\log (G(\tau))$ provides a clearer short-horizon Lyapunov proxy for the modulus dynamics.

\begin{figure}[H]
\centering

\subfloat[Shape of the modulus error curves for ref002.]{
  \includegraphics[width=0.45\textwidth]{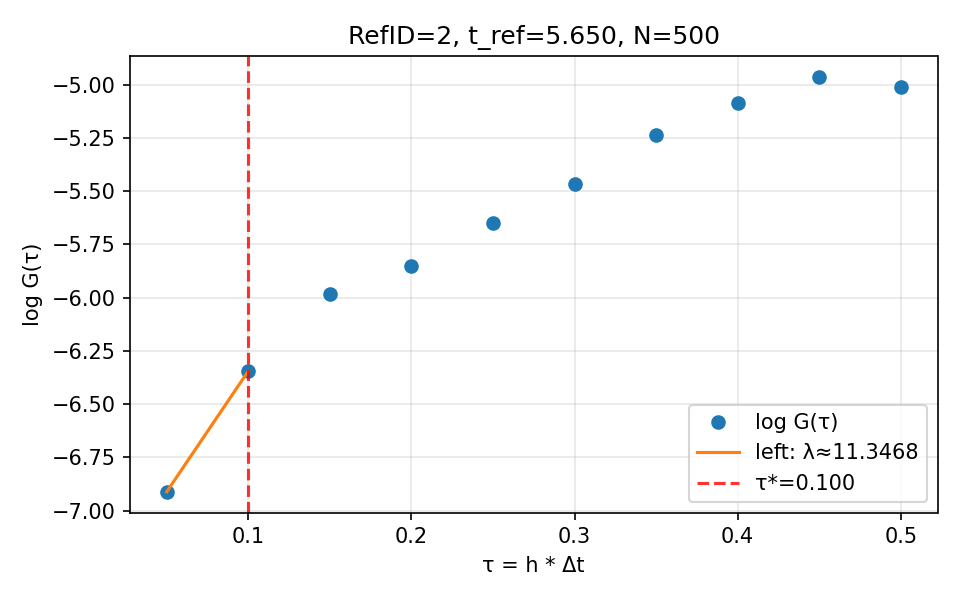}
  \label{fig:figure15}
}
\hfill
\subfloat[Shape of the modulus error curves for ref032.]{
  \includegraphics[width=0.45\textwidth]{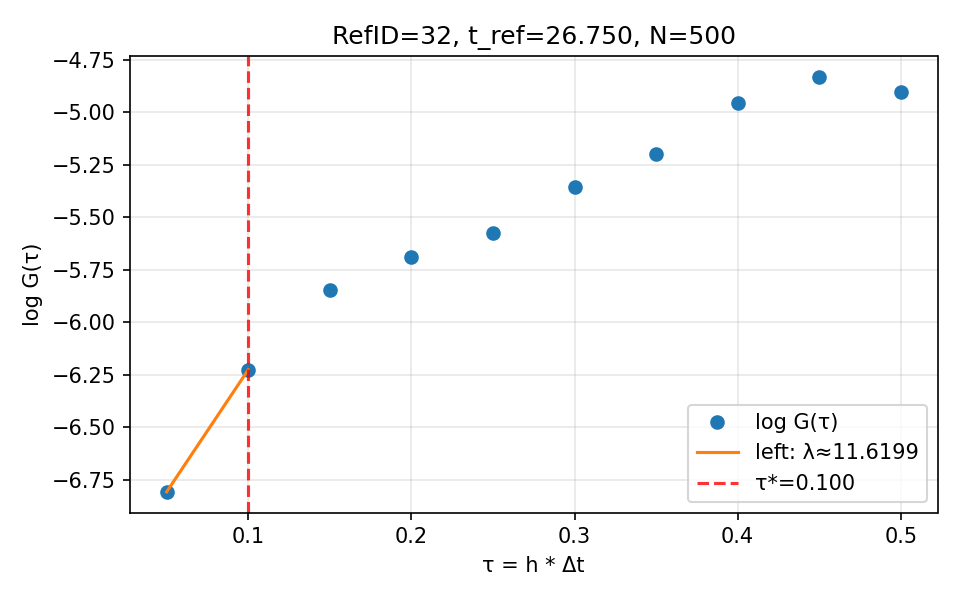}
  \label{fig:figure16}
}

\caption{Examples of modulus-based kNN prediction--error curves
$l(h)=\log G(h)$ and their MAE--optimal linear fits for two reference states.
Blue markers show the empirical profile as a function of $\tau=h\Delta t$,
while the solid line indicates the selected short--horizon fit.}
\label{fig:fig15fig16}
\end{figure}

\textbf{\textsl{Time series of modulus-based Lyapunov exponents.}}
Collecting the short-horizon slopes over all reference times yields the
modulus-based Lyapunov proxy $\hat{\lambda}_{1}^{(\mathrm{mod})}(t_{\text{ref}})$, shown in
Figure~\ref{fig:figure17}. The resulting series fluctuates within a
well-defined range of values comparable in magnitude to the FTLE-based estimates,
confirming that the modulus carries a coherent signal of local instability.
However, the fluctuations are less clearly quasi-periodic than in the FTLE time
series: although some broad oscillatory structure is visible, the modulation is
partly irregular and does not follow a simple near-sinusoidal pattern. This
suggests that the modulus-based kNN estimator captures the overall level of local
stretching and contraction but remains more sensitive to local variability in the
micro-ensemble than the FTLE construction.\\
\begin{figure}[htbp]
  \centering
  \includegraphics[width=\textwidth]{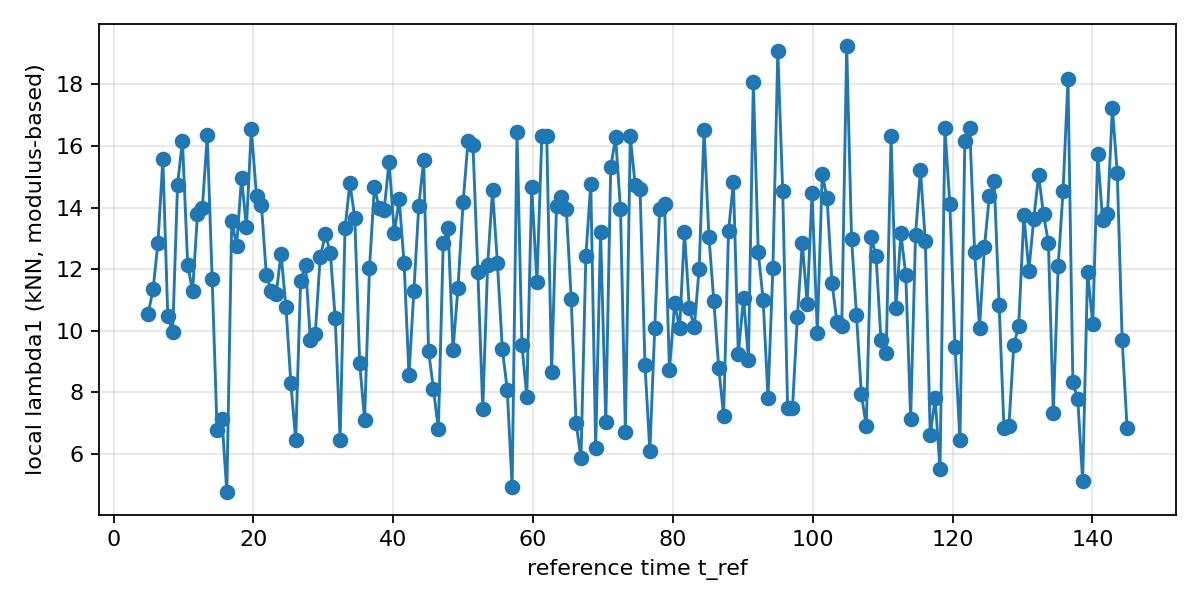}
  \caption{Local modulus-based Lyapunov proxy $\hat{\lambda}_{1}^{(\mathrm{mod})}(t_{\text{ref}})$
obtained from kNN prediction errors on the scalar modulus $r(t)$. The values remain
within a finite, physically reasonable range and fluctuate along the trajectory,
but the oscillations are less clearly periodic than in the FTLE-based estimate,
reflecting a combination of dynamical modulation and local ensemble variability.}
  \label{fig:figure17}
\end{figure}

\subsection{Summary and comparison of the data-driven proxies}
\indent Table \ref{tab:3method compact} summarizes the main similarities and differences between the three numerical Lyapunov-proxy techniques used in Sections 5. The FTLE method (Method I) relies on micro-ensemble divergence and provides the most reliable estimate because it captures the true exponential separation of nearby trajectories. The kNN full-state method (Method II-A) operates in the full eight-dimensional state space and exhibits strong noise sensitivity, which leads to unstable slopes and weak interpretability. In contrast, the modulus-based kNN method (Method II-B) applies the prediction-error approach to the scalar state modulus, yielding smoother curves and more stable short-horizon growth. Overall, the FTLE method gives the best performance, the full-state kNN method is the least reliable, and the modulus-based kNN method provides a good compromise between dimensionality reduction and numerical stability.
\begin{table}[h!]
\centering
\caption{Compact comparison of the three Lyapunov--proxy methods.}
\label{tab:3method compact}
\begin{tabular}{l|c|c|c}
\hline
\textbf{Feature}
& \textbf{FTLE}
& \textbf{kNN Full}
& \textbf{kNN Modulus} \\
\hline

Dim.
& 8-D
& 8-D
& 1-D \\

Data
& Main + micro
& Micro only
& Micro only \\

Measure
& $d_{\mathrm{geo}}$
& $G(\tau)$
& $G(\tau)$ \\

Range
& $0.05\!\to\!1.7$
& $0.022\!\to\!0.026$
& $0.001\!\to\!0.007$ \\

Log-range
& $-2.95\!\to\!0.54$
& $-3.86\!\to\!-3.64$
& $-6.91\!\to\!-4.96$ \\

Stability
& High
& Low
& Med--high \\

Curve
& Clean rise
& Noisy
& Smooth \\

Interp.
& Strong
& Weak
& Good \\

Overall
& \textbf{Best}
& \textbf{Poor}
& \textbf{Good} \\
\hline
\end{tabular}
\end{table}

\section{Conclusion}
This paper presented a unified analytical and data-driven framework for synchronization and local stability assessment in fractional-order complex-valued BAM neural networks. Using fractional Lyapunov theory together with Mittag--Leffler functions, we derived sufficient conditions guaranteeing global Mittag--Leffler (GML) synchronization of the drive--response pair under a linear error-feedback controller, and obtained a conservative explicit time-to-tolerance estimate from a standard Mittag--Leffler bound. The numerical experiments corroborated the theoretical analysis, showing rapid decay of synchronization errors and consistent regulation of both real and imaginary components of the BAM states.

Beyond the model-based guarantees, we introduced two practical Lyapunov proxies that operate directly on trajectory data. The micro-ensemble FTLE estimator captures short-horizon divergence through the geometric-mean growth of perturbations and reveals oscillatory transient-instability patterns typical for fractional-order dynamics. In parallel, we implemented a kNN prediction-error proxy in two complementary regimes: \text{Method II-A} (full-state kNN), which forecasts the complete real-imaginary stacked state in $\mathbb{R}^{8}$ and preserves component coupling information, and \text{Method II-B} (modulus kNN), which forecasts a scalar magnitude and provides a lower-dimensional, typically smoother error-growth signature. Both kNN variants consistently reflected the stabilizing impact of the designed controller and agreed qualitatively with the FTLE-based diagnostics, thereby strengthening the empirical validation of the proposed synchronization framework.

Overall, the combination of fractional-order synchronization theory, explicit convergence-rate estimates, and trajectory-driven Lyapunov proxies bridges model-based and data-driven perspectives for complex-valued fractional-order neural networks, with relevance to secure communications and nonlinear signal processing where reliable synchronization and stability diagnostics are essential.

\textbf{Future work.} Several directions appear promising: (i) extending the analysis to large-scale and heterogeneous BAM architectures (including network-of-networks settings) and to uncertain parameters, disturbances, and time-varying delays; (ii) developing adaptive or gain-scheduled controllers informed by online FTLE/kNN indicators to improve transient performance while maintaining theoretical guarantees; (iii) studying robustness of the micro-ensemble and kNN proxies under measurement noise and limited ensemble size, including automated selection of embedding and neighborhood parameters; and (iv) benchmarking against additional data-driven stability tools (e.g., kernel methods or Koopman-based predictors) and validating the approach on real-world signals or hardware-in-the-loop implementations. The proposed framework can be extended to other fractional-order systems, including chaotic oscillators and memristive neural networks, due to its reliance on general Lyapunov-based analysis and LMI techniques. Such extensions may require suitable modifications to accommodate system-specific nonlinearities and will be considered in future work.

\section*{Acknowledgments}
This research was supported by the Russian Science Foundation (grant no. 22-11-00055-P, \url{https://rscf.ru/en/project/22-11-00055/}, accessed on 10 June 2025).
The authors also wish to express sincere gratitude to the editor and reviewers for their insightful and constructive feedback on the manuscript.

\section*{Data availability}
The data supporting the findings of this study are included within this article.

\section*{Conflict of interest}
The authors declare that there are no conflicts of interest related to the publication of this article.

\end{document}